\documentclass[twoside,11pt]{article}

\usepackage{jmlr2e}
\usepackage{caption}
\usepackage{color}
\usepackage{hyperref}
\usepackage{algorithm}
\usepackage[noend]{algorithmic}

\def\min{\qopname\relax n{min}}
\def\max2{\qopname\relax n{max2}}
\def\max{\qopname\relax n{max}}

\def\Pr{\qopname\relax n{\mathbf{Pr}}}

\newcommand{\EE}{\mathbb{E}}
\newcommand{\II}{\mathbb{I}}
\newcommand{\RR}{\mathbb{R}}
\newcommand{\NN}{\mathbb{N}}

\def\G{\mathcal{G}}

\def\M{\mathcal{M}}

\def\S{\mathcal{S}}

\def\T{\mathcal{T}}

\def\X{\mathcal{X}}

\def\al{\bm{\alpha}} 
\def\x{\bm{x}} 
 
\def\y{\bm{y}}

\def\s{\bm{s}} 
 
\def\p{\bm{p}}

\def\y{\bm{y}}

\newenvironment{lp*}{\begin{equation*}  \begin{array}{lll}}{\end{array}\end{equation*}}

\usepackage[utf8]{inputenc}
\usepackage{tikz}
\usepackage{stmaryrd}
\usepackage{verbatim}
\usepackage{bm}
\usepackage{dsfont}
\usepackage{url}
\usepackage{array}
\usepackage{mathtools}
\usepackage{enumitem}
\usepackage{hyperref}

\usepackage{diagbox}

\hypersetup{
    colorlinks=true,
    linkcolor=blue,
    filecolor=magenta,      
    urlcolor=cyan,
}
\usepackage{soul}

\usepackage{cancel}

\usepackage{graphicx}
\usepackage{epstopdf}

\AppendGraphicsExtensions{.gif}

\usepackage{xcolor}
\definecolor{darkspringgreen}{rgb}{0.09, 0.45, 0.27}

\newcommand{\game}{{competing content creation game}}

\definecolor{grey}{rgb}{0.8,0.8,0.8}

\newtheorem{theorem}{Theorem} 
\newtheorem{proposition}{Proposition}
\newtheorem{lemma}{Lemma}
\newtheorem{corollary}{Corollary}

\newtheorem{definition}{Definition}

\jmlrheading{1}{2021}{1-48}{4/00}{10/00}{meila00a}{Fan Yao, Chuanhao Li, Denis Nekipelov, Hongning Wang and Haifeng Xu}

\ShortHeadings{How Bad is Top-$K$ Recommendation under Competing Content Creators?}{How Bad is Top-$K$ Recommendation under Competing Content Creators?}
\firstpageno{1}

\begin{document}

\title{How Bad is Top-$K$ Recommendation under Competing Content Creators?}

\author{\name Fan Yao$^1$ \email fy4bc@virginia.edu
       \AND
       \name Chuanhao Li$^1$ \email cl5ev@virginia.edu 
       \AND
       \name Denis Nekipelov$^2$ \email dn4w@virginia.edu 
       \AND
       \name Hongning Wang$^1$ \email hw5x@virginia.edu 
       \AND
       \name Haifeng Xu$^3$ \email haifengxu@uchicago.edu \\ \\ 
       \addr $^1$Department of Computer Science, University of Virginia, USA \\
       \addr $^2$Department of Economics, University of Virginia, USA\\
       \addr $^3$Department of Computer Science, University of Chicago, USA
    }


\maketitle

\begin{abstract}
Content creators compete for exposure on recommendation platforms, and such strategic behavior leads to a dynamic shift over the content distribution. However, how the creators' competition impacts user welfare and how the relevance-driven recommendation influences the dynamics in the long run are still largely unknown.

This work provides theoretical insights into these research questions. We model the creators' competition under the assumptions that: 1) the platform employs an innocuous top-$K$ recommendation policy; 2) user decisions follow the Random Utility model; 3) content creators compete for user engagement and, without knowing their utility function in hindsight, apply arbitrary no-regret learning algorithms to update their strategies. We study the user welfare guarantee through the lens of Price of Anarchy and show that the fraction of user welfare loss due to creator competition is always upper bounded by a small constant depending on $K$ and randomness in user decisions; we also prove the tightness of this bound. Our result discloses an intrinsic merit of the myopic approach to the recommendation, i.e., relevance-driven matching performs reasonably well in the long run, as long as users' decisions involve randomness and the platform provides reasonably many alternatives to its users.
\end{abstract}

\section{Introduction}

\begin{quote}
   \textit{``(Producers) are led by an invisible hand to make nearly the same distribution of the necessaries of life... thus without intending it, without knowing it, advance the interest of the society.''
   } 
    
    \hspace{35mm} --- Adam Smith, \emph{The Theory Of Moral Sentiments}, 1759. 
\end{quote}

Online recommendation platforms such as Instagram and YouTube have become prevalent in our daily life \citep{bobadilla2013recommender}. At the core of those platforms is a recommender system (RS) designed to match each user with the most relevant content based on predicted relevance. Such a practice, often referred to as the top-$K$ RS, is believed to improve user satisfaction and has served as a rule-of-thumb principle in both academia and industry for decades \citep{konstan1997grouplens,koren2009matrix,he2017neural}. 

Until recently, the community came to realize that users' utilities cannot be maximized unilaterally due to the potential strategic behaviors of content creators \citep{qian2022digital}. Because the content creators' utilities are directly tied to their content's exposure, they are motivated to adaptively maximize their own utilities. This leads to competition that may potentially be harmful to social welfare (defined as the total user satisfaction/engagement) \citep{fleder2009blockbuster}. For example, consider a scenario where the user population contains a large group of sports fans and a small group of travel enthusiasts. Social welfare is maximized when the available content for recommendation covers both topics. However, one possible equilibrium of the competition is that all creators post homogeneous sports content when the gain from creating niche content cannot compensate for the utility loss caused by abandoning the exposure from the majority of users.   It is thus urgent to understand in the long run how bad the social welfare loss could be under strategic content creators driven by a top-$K$ RS. 

In this work, we propose the \textit{\game{}} to model the impact of the creators' competition on user engagement in a top-$K$ RS.
We measure the social welfare guarantee through the lens of Price of Anarchy (PoA) \citep{koutsoupias1999worst}, which quantifies the inefficiency of selfish behavior by the ratio between the worst-case welfare value of the game's equilibrium and that of an optimal outcome. Some previous works touched upon this question under different competition models, and their answers are all pessimistic. For example, \citet{ben2018game} noticed that the PoA of social welfare under the RS implemented by a Shapley mediator is unbounded. \citet{ben2019recommendation} studied a competition model in $1$-dimensional space and showed that the PoA under the top-$1$ matching principle could be as bad as a constant $2$. These negative results are either based on a deterministic user choice model or assume creators compete for the shares of \emph{content exposure}. We overturn these pessimistic conclusions by showing that the PoA induced by a top-$K$ RS is at most $1 +  O( \frac{1}{ \log K})$  when  (1)  $K>1$,  (2) user choices have mild stochastic noises, and (3)   creators are incentivized to compete for \emph{user engagement} instead of content exposure. We also prove its  tightness   by analyzing a lower-bound instance.   Thus   an RS under these assumptions will approach the optimal efficiency  (i.e., PoA ratio approaches $1$) when    $K$ grows, though at a relatively slow rate of $1/\log K$. 
 Notably, our PoA upper bound also holds in dynamic settings where creators gradually learn to improve their strategies in an online fashion. Extensive synthetic and real-world data based simulations also support these theoretical findings. Overall, our results robustly demonstrate that content creation competitions are efficient under properly set incentives. This   echoes the famous insight of Adam Smith, as cited at the beginning of the section,  about the market's ``invisible hand'' on driving socially efficient  production of goods. 


Our results rely on three key assumptions, all of which find their roots in recommendation literature and practice. First, on the platform side, we assume the top-$K$ RS is based on a relevance function that best predicts user satisfaction if recommended content is consumed. To simplify our setting, we assume the true relevance function is known to the RS, since a tremendous amount of research has been spent on this aspect \citep{bobadilla2013recommender,konstan1997grouplens,koren2009matrix,he2017neural} and the goal of our study is not to improve its estimation. 
Second, on the user side, we employ the well-established \emph{Random Utility (RU) model} \citep{manski1977structure} to specify the distributional structure of a user's choices and resulting utility when presented with a list of recommendations. The RU model has been widely adopted and found its success in marketing research to model consumer choices \citep{baltas2001random}. 
Third, on the creator side, we assume that their utilities collected from matching their content with a user are proportional to the user's utility, as it is a common practice by platforms to set revenue sharing with content creators proportional to the user's satisfaction or engagement \citep{meta_experiment,savy2019,youtube2023,tiktok2022}. When we move on to the dynamic setting where the creators do not have oracle access to their utility functions, we allow creators to adopt arbitrary no-regret learning algorithms, which cover a variety of rational learning behaviors. 

An additional merit of our result is that our characterization of social welfare does not require any stability property (e.g., the existence or uniqueness of Nash equilibria) of the competition, i.e., it applies to the content distribution induced by any reasonable sequence of content creation strategies, regardless of whether or not the strategy sequence converges. While previous works \citep{ben2018game,ben2017shapley,ben2019recommendation} strive to establish a unique pure Nash equilibrium (PNE) guarantee in similar game-theoretic settings, we do not consider such a stability requirement crucial for the system's design for two reasons. 
First, as demonstrated in previous works \citep{ben2018game,ben2017shapley,ben2019recommendation}, the unique PNE 
is not guaranteed for top-$K$ RS in general. However, our main result indicates the social welfare can be ensured under top-$K$ RS regardless of the existence of a PNE, thus eliminating the need for tradeoffs between stability and the complexity of the recommendation algorithm. Second, even in cases where a PNE does exist, it is unclear how creators can achieve such a stable outcome in practice. For instance, \citet{ben2018game,ben2017shapley} showed that creators following best-response dynamic \citep{monderer1996potential} converge to the unique PNE in their game-theoretical setting, but it requires creators to have oracle access to their utility functions, which is unrealistic as creators can only evaluate the utility of their taken actions (i.e., bandit feedback) in practice. 
Therefore, we refrain from discussing stability and instead focus on characterizing the average social welfare under the evolving strategies of creators.

\section{Related Work}

The theoretical studies of content creators' strategic behavior under the mediation of an RS date back to the seminal works from \citet{ben2017shapley,ben2018game}, where they extended the game setting in search and ranking systems \citep{ben2015probability,raifer2017information} and proposed an RS based on Shapley value that leads to the unique PNE and several fairness-related requirements. However, they showed that the social welfare under the proposed Shapley mediator could be arbitrarily bad. 


Another line of work studies the RS with strategic content creators under the Hotelling's spatial competition framework \citep{hotelling19291929}. First introduced by \citet{hotelling19291929}, Hotelling's model studied two restaurants trying to determine their locations to attract customers who are evenly distributed on the segment $[0,1]$. The Nash equilibrium (NE) of the resulting game is that both restaurants locate at the center, known as the ``principle of minimum differentiation". Recently, \citet{shen2016hotelling} proposed a variant of Hotelling's competition in which each player has its attraction region, and they showed that the social welfare at the NE is half of the maximum possible social welfare in the worse case. We show that their game settings are special cases of our proposed \game{} in Appendix \ref{app:connection}, and thus our main result directly implies their bound. A more closely related work is from \citet{ben2019recommendation}, where they introduce the RS into the competition as a mediator who directs users to facilities. They studied mediators with different levels of intervention and proposed a limited intervention mediator with a good trade-off between social welfare and intervention cost. Interestingly, their game setting under a no-intervention mediator also turns out to be a special case of ours. We also note that the problem settings and theoretical discussions in both \citep{shen2016hotelling} and \citep{ben2019recommendation} are limited to pure strategy in 1-dimensional cases with a distance-induced user utility function, while our model and result apply to arbitrary dimensions and a generic form of user utility functions. 


Two recent works \citep{hron2022modeling,jagadeesan2022supply} studied the supply-side competition where the creators' strategy space is high dimensional. Their models assume creators directly compete for user exposure without considering the role of an RS. They focused on the characterization of NE and the identification of conditions under which specialization among creators' strategies may occur. In contrast, we study the social welfare under the impact of a top-$K$ RS without being limited to the existence of NE, and our result applies to general user utility functions.


Our user decision model stems from the RU model \citep{baltas2001random} in econometrics, which explains how an individual makes choices among a discrete set of alternatives. In the RU model, the utility that a decision maker could obtain from alternative $j$ is decomposed into $U_j=V_j+\epsilon_j$, where $V_j$ is the known parameterized part, and $\epsilon_j$ is the unknown stochastic part. The observed choice is then given by the alternative with the maximum utility. It is shown that if the unobserved stochastic utility follows the extreme value distribution (i.e., Gumbel distribution), then the choice probability is given by the logit formula, i.e., $P_j \propto \exp(V_j)$ \citep{luce1965preference}. In our work, we apply the RU model to explain how a typical user allocates her attention across the recommended list.

Finally, to analyze the equilibrium efficiency of the \game{}, we employed the standard framework of the price of anarchy (PoA). This originates from the seminal work of \cite{koutsoupias1999worst} and has since led to an extensive literature on understanding the efficiency of numerous strategic games. Our discussion by no means can do justice to this rich literature; here, we only mention the few works that are closely related to ours. Since NE is not guaranteed to exist in our problem with non-continuous agent utilities \citep{hron2022modeling}, it is thus crucial for us to consider a solution concept that is weaker than NE and thus to prove a stronger PoA bound. Specifically, we consider coarse correlated equilibrium (CCE). The PoA for CCE is first studied by \cite{blum2008regret}, who considered the efficiency of a dynamic setup with no-regret learners and coined the new notion of the price of total anarchy, which turns out to be equivalent to the PoA bound for CCE. This is precisely the question we want to address, but the structure of our new \game{} is significantly different from the games they studied, such as Hotelling's game on a graph and the valid utility game of  \cite{vetta2002nash}. Thus their techniques are not readily applicable to our problem. We instead employed a recent framework of \cite{roughgarden2015intrinsic} using the smoothness argument. It is well-known that this framework can yield strong PoA bound applicable to CCE. However, the bounds obtained by this powerful framework are usually not tight; so far, it is only known that it yields tight PoA bounds for linear cost congestion games \citep{christodoulou2005price}, second price auctions \citep{christodoulou2008bayesian}, and the valid utility games \citep{vetta2002nash}. Interestingly, We show that the smoothness argument also yields a tight PoA bound for our \game{} and thus register an additional member to this important list of games.     

\section{A Model of Content Creator Competition}

In this section, we formalize the \game{}. The game $\G$ is defined by a tuple $(\{\S_i\}_{i=1}^n, \X, \sigma, \beta, K)$ with the following ingredients:
\begin{enumerate}
    \item A finite set of users $\X=\{\x_j\in \RR^d\}_{j=1}^m$, and a set of players (i.e., content creators\footnote{We use these two terms interchangeably when there is no ambiguity. }) denoted by $[n]=\{1,\cdots,n\}$. Each player $i$ can take an action $\s_i$, often referred to as a \emph{pure strategy} in game-theoretic literature, from an action set\footnote{The action sets are not assumed to be finite and thus can be continuous.} $\S_i \subset \RR^d$. $\s_i$ can be understood as the embedding for the \emph{type} of content that creator $i$ can produce.
    Let $\S=\prod_{i=1}^n \S_i$ denote the space of joint strategies. As a convention, for any $\s=(\s_1,\cdots,\s_n) \in \S$, we use $\s_{-i}$ to denote the joint strategy $\s$ excluding $\s_i$. Moreover, we use $\al_i \in \Delta(\S_i)$ to denote a mixed strategy of player $i$, which is a probability measure with support $\S_i$. Similarly, $\al\in\Delta(\S)$ is used to represent a (possibly correlated) joint strategy profile distribution over all players.
    \item A relevance function $\sigma(\s, \x): \RR^d\times \RR^d \rightarrow \RR_{\geq 0}$ which measures the \emph{relevance} between a user $\x \in \X$ and content $\s$. Without loss of generality, we normalize $\sigma$ to $[0, 1]$, where $1$ suggests perfect matching. We focus on modeling the strategic behavior of creators and thus abstract away the estimation of $\sigma$. 
    \item Recommendation policy: given a joint strategy $\s=(\s_1,\cdots,\s_n)\in\S$ for all players, for each user $\x_j$, the RS first calculates the relevance scores $\{\sigma(\s_i,\x_j)\}_{i=1}^n$ over all available content and then generates $\T_j(\s;K)$, the subset of $\s$ containing the top-$K$ recommendations for user $j$. Formally, 
    \begin{equation}\label{eq:Tj}
        \T_j(\s;K) = \{\s_{l_i}|i=1,\cdots,K\},
    \end{equation}
    where $(l_i)_{i=1}^n$ is a permutation of $[n]$ such that $\sigma(\s_{l_1},\x_j)\geq \sigma(\s_{l_2},\x_j)\geq \cdots\geq \sigma(\s_{l_n},\x_j)$.\footnote{When $(l_i)_{i=1}^n$ is not unique, $\T_j(\s;K)$ can be the top-$K$ truncation of any such permutation with equal probability.} 
    \item User utility and choice model: we employ the widely adopted random utility (RU) model to capture users' utility and choices of recommendations. Formally, the RU model assumes that the utility for user $\x_j$ to consume content $\s_i$ is $\sigma(\s_i, \x_j)+\varepsilon_i$, where 
    $\varepsilon_i$ is a noise term containing any additional uncertainty that cannot be captured by the RS's prediction $\sigma(\s_i, \x_j)$  (e.g., user's mood at that moment). The RU model assumes that $ \{ \varepsilon_i \} $ are  i.i.d. random, which are often assumed to follow the \emph{Gumbel distribution} with cumulative distribution function Gumbel$(\mu, \beta) = e^{-e^{-\frac{x-\mu}{\beta}}}$.\footnote{There are many natural reasons to use the Gumbel noise model.  This noise model is nearly indistinguishable from a Gaussian distribution empirically, but has slightly thicker tails, allowing for more aberrant user behavior. The RU model with Gumbel noise is also known as the multinomial logit model \citep{mcfadden1974measurement}. It deeply connects to the discrete choice model \citep{mcfadden1984econometric}, quantal response equilibrium to capture bounded rational behaviors \citep{mckelvey1995quantal}, and entropy regularizer for optimizing randomized strategies \citep{ling2018game}. } We further assume $\varepsilon_i$ is zero mean, thus implying $\mu = -\beta\gamma$ where $\gamma\approx 0.577$ is the Euler–Mascheroni constant. The variance of Gumbel$(-\beta\gamma, \beta)$ is $\frac{\pi \beta}{\sqrt{6}}$  and the parameter $\beta$ measures the noise level.  
 
    Upon receiving the recommended list $\T_j(\s;K)$, user $j$ chooses $i_j^*\in \T_j(\s;K)$ that maximizes her utility:  
    \begin{equation}\label{eq:best-user-response}
        i_j^*=\arg\max_{\s_i\in \T_j(\s;K)} \{\sigma(\s_i, \x_j)+\varepsilon_i\}. 
    \end{equation}
    Note that $i^*_j$ is random, with randomness inherited from $\{ \varepsilon_i \}$. Consequently,  user $j$ derives the following expected utility $\pi_j$ from consuming the selected content \begin{equation}\label{eq:user_utility_original}
      \pi_j(\s) \triangleq  \EE_{\{ \varepsilon_i \} }\left[\max_{\s_i\in \T_j(\s;K)} \{\sigma(\s_i, \x_j)+\varepsilon_i\} \right].  
    \end{equation}
     
    \item Player utilities: following the convention, we assume that each player $i$'s expected utility is the sum of the utilities from users that $i$ served, i.e.,  
    \begin{equation}\label{eq:player_utility_original}
            u_i(\s)=\sum_{j=1}^m \EE[\sigma(\s_i, \x_j)+\varepsilon_i|\x_j \shortrightarrow \s_i]\cdot \Pr[\x_j  \shortrightarrow \s_i],
    \end{equation}
    
    where ``$\x_j\shortrightarrow \s_i$'' denotes the event $i_j^*$ defined in   \eqref{eq:best-user-response}  equals $i$. Elegantly, $\Pr[\x_j \shortrightarrow \s_i] \propto e^{\beta \sigma(\x_j, \s_i)}$ for any $i\in \T_j(\s;K)$ \citep{mcfadden1974measurement} and  $\Pr[\x_j \shortrightarrow \s_i]=0$ if $i\notin \T_j(\s;K)$.

    \item Social welfare: the social welfare function is defined as the total utilities from all the users: 
    \begin{equation}\label{eq:welfare_original}
        W(\s) = \sum_{j=1}^{m}\pi_j(\s).
    \end{equation}
    Note that under the player utility function \eqref{eq:player_utility_original}, we have $W(\s) = \sum_{i=1}^{n}u_i(\s)$. That is, the social welfare is also the total utility of players. 
\end{enumerate}

 
We remark that the player $i$'s utility defined in \eqref{eq:player_utility_original} depends on not only the proportion of users matched with $i$, but also the user's engagement  reflected in the term $\EE[\sigma(\s_i, \x_j)+\varepsilon_i|\x_j \shortrightarrow \s_i]$.  
This differs crucially from the settings in Hotelling's competition \citep{hotelling19291929} and its recent applications to recommender systems \citep{shen2016hotelling,ben2019recommendation,hron2022modeling,jagadeesan2022supply}, where players' utilities are set to the total \emph{user exposure}, i.e., total number or proportion of user visits (regardless of how satisfied the users are with the recommendations). Both metrics have been widely used in current industry practice to reward creators \citep{meta_experiment,savy2019}. In this paper, we primarily consider user engagement (i.e., the previously less studied case) as the creator's utility, and in Section \ref{sec:imp_poa} we will compare it with the \emph{user exposure} metric to highlight their different impact.   

\vspace{2mm}
\noindent 
\textbf{Our research question and equilibrium concept. } 
We are particularly interested in quantifying the average social welfare when creators learn to update their strategies adaptively. Specifically, we consider the repeated form of a \game{} played by $n$ creators over a period of time $T$. At each time $t$, each creator chooses an action, observes the utility induced by all creators' strategies at that round, and uses the feedback to adjust their subsequent actions. Naturally, creators aim to optimize their accumulated utility over the course of interactions. However, in real-world online recommendation platforms, creators can only evaluate the utility of their chosen actions and have to gradually learn their optimal strategies through trial and error with such limited information (i.e., bandit feedback). A natural notion for capturing the ``reasonable'' learning behavior under such environment is \emph{no regret}. The (external) regret $R_i(T)$ for player $i$ is defined as the difference between her optimal utility in hindsight and the realized accumulated utilities, i.e.,
\begin{equation}\label{def:regret}
\small R_i(T) =   \max_{\s'_i}  \sum_{t=1}^T\EE_{\s_{-i}\sim\al^t_{-i}}[u_i(\s'_i,\s_{-i})] - \sum_{t=1}^T\EE_{\s\sim\al^t}[u_i(\s)]  
\end{equation}
\normalsize
where $\al^t=\prod_{i=1}^n \al_i^t$ denotes the joint-strategy distribution at time $t$. 
Player $i$'s learning has no regret if $R_i(T) = o(T)$, or equivalently, the average regret $R_i(T)/T \to 0$ as $T$ goes to infinity. Note that such no-regret algorithms exist since any no-regret adversarial online learning algorithm (e.g., Exp3 in bandit literature \citep{auer2002nonstochastic}) guarantees no regret in such a multi-agent learning setup.


To characterize the outcome under no-regret learning players, we focus on an equilibrium concept termed coarse correlated equilibrium (CCE), as it is well known that the empirical action distribution of any no-regret playing sequence in a repeated game converges to its set of CCEs \citep{blum2008regret}. The formal definition of CCE is as follows: 
\begin{definition}\label{def:CCE}
    A \emph{coarse correlated equilibrium} (CCE) is a distribution $\al$ over the space of joint-strategy profile $\S$ such that for every player $i$ and every action $\s'_i\in\S_i$, 
    \begin{equation}
           \EE_{\s\sim\al}[u_i(\s)] \geq \EE_{\s\sim\al}[u_i(\s'_i,\s_{-i})].
    \end{equation}
\end{definition}

Thanks to the nice connection between no-regret dynamics and CCE, we first establish the welfare guarantee for CCE in Section \ref{sec:main_result} and then extend it to account for the accumulated welfare induced by repeated plays in Section \ref{sec:imp_poa}. 

We also note that the concept of CCE is particularly useful for two additional reasons. First, CCE always exists in any finite games (thus in our game), hence eliminating the necessity to address the existence of Nash equilibrium (NE), perhaps the most celebrated solution concept, as in previous research \citep{hron2022modeling}. In fact, when the action sets are continuous, the existence of NE (either pure or mixed) cannot be guaranteed in our game as the player utility function defined in  \eqref{eq:player_utility_original} is not continuous. This is an inherent challenge of the problem, as any change in $\sigma(\s, \x)$ may result in a different top-$K$ recommendation list $\T_j(\s;K)$, leading to dramatically different player utilities. Similar challenges and the non-existence of mixed NE have also been observed by \citet{hron2022modeling}, though their utility model and research questions differ from ours. Second, even in situations where NE exists, it is more realistic to assume that players eventually achieve some CCE rather than NE due to various criticisms about NE, including the computational concerns \citep{daskalakis2009complexity} of NE.



 \section{The Price of Anarchy Analysis} 

We analyze the social welfare of any top-$K$ RS under any possible CCE; or more specifically, \emph{how bad can the welfare possibly be due to the competition among self-interested content creators} -- compared to the \emph{idealized} non-strategic situation in which the platform can ``dictate'' all creators' content choices and thus globally optimize the welfare function \eqref{eq:welfare_original}. This can be captured by the celebrated concept of the Price of Anarchy (PoA) \citep{koutsoupias1999worst}. As its name indicates, PoA captures the welfare inefficacy due to players' selfish behavior. Our main result in this section is a comprehensive characterization of the PoA of \game{}s. 
\begin{definition}[PoA under CCE]
Define the price of anarchy of a game $\G$ as 
\begin{equation}\label{eq:poa_def}
        PoA(\G)=\frac{\max_{\s\in\S}W(\s)}{\min_{\al\in CCE(\G)}\EE_{\s\sim \al}[W(\s)]},
\end{equation}
where CCE$(\G)$ is the set of CCEs of $\G$.  
\end{definition}

By definition, PoA$(\G)\geq 1$ always holds and larger values indicate worse welfare. Our choice of the CCE concept leads to the strongest possible welfare guarantee in the sense that any upper bound of PoA under CCE also trivially holds for the PoA under refined solution concepts such as correlated equilibrium (CE), PNE or mixed NE (if they exist), since these are all CCEs as well. Unless otherwise emphasized, any PoA in this paper refers to the PoA under CCE.

\subsection{Matching PoA Upper and Lower Bounds}\label{sec:main_result}

Our main theoretical findings are an upper and lower bound for the PoA, which match with each other and thus demonstrates the tightness of our analysis. We first present the upper bound as follows.

\begin{theorem}\label{thm:PoA}
The PoA of any \game{} instance $\G$ with parameter  $\beta\geq0$ and $K\geq 1$  satisfies 
\begin{equation}\label{eq:PoA_main}
    PoA(\G) < 1+ \frac{1}{c(\beta,K)},
\end{equation}
where $c(\beta, K)$ is defined as

\begin{align}\label{eq:cbetaK_def}
c(\beta, K)= \frac{(b+1)\log(b+K)}{(b+K)(\log(b+K)-\log K)}, b=e^{\frac{1}{\beta}}-1.
\end{align}    

\end{theorem}
The proof of Theorem \ref{thm:PoA} is intricate and thus the detailed arguments are relegated to Appendix \ref{app:poa_proof}. The primary challenge in the proof is to analyze various smoothness properties of the welfare and players' utility functions, especially how the welfare function changes after excluding any player $i$'s participation. In Section \ref{subsec:proof-sketch}, we highlight some of the noteworthy properties of the welfare function, including its submodularity, which we develop en route to proving Theorem \ref{thm:PoA} but is also of independent merit towards understanding the \game{}.   
 
The format of $c(\beta, K)$ may not be intuitive enough for the readers to appreciate the derived PoA upper bound. We thus provide the following observations, which reveal various properties of $c(\beta, K)$ aiding the interpretation of \eqref{eq:PoA_main}:
\begin{enumerate}
    \item For any $\beta>0$ and $K\geq 1$, we have $c(\beta, K)\geq 1$ and thus PoA$(\G) < 2$ always holds.

    \item $c(\beta, K)=1$ if and only if $K=1$ or $\beta \rightarrow 0$.
    
    \item Fix any $\beta>0$, $c(\beta, K)$  monotonically increases in $K$; similarly, fix any $K\geq 1$, $c(\beta, K)$ monotonically increases in $\beta$.
    
    \item For sufficiently large $\beta$ and $K$, $c(\beta, K) \approx (1+\beta)\log K$ asymptotically, and therefore 
    \begin{equation}\label{eq:PoA_upper}
        PoA(\G) < 1+\frac{1}{(1+\beta)\log K}.
    \end{equation}
\end{enumerate}

Based on these observations, Theorem \ref{thm:PoA} has multiple interesting and immediate implications. First, the welfare loss under any CCE is at most half in any situation, as the PoA is always upper bounded by 2. The second and third facts above show that such worst-case PoA occurs and only occurs when users' choices are made in a ``hard" manner: either the RS dictates the user's choice by setting $K=1$ or the randomness in users' choices is extremely low (i.e., $\beta \rightarrow 0$). Note that in the latter case, the user will only consume the most relevant content (i.e., the top-ranked content) due to small decision randomness.   

Second, the welfare guarantee improves as either $K$ increases (i.e., more items are recommended) or $\beta$ increases (i.e., users' choices have more randomness). Welfare improvement in the latter situation is intuitive because when supplied with multiple items, the user can pick the content with large $\varepsilon_i$ (i.e., the reward component that is not predictable by the RS) to gain utility. These together reveal an interesting operational insight that when the RS cannot perfectly predict user utility (i.e., $\beta > 0$), providing more items can help improve social welfare. This justifies top-$K$ recommendation and the necessity of diversity in recommendation \citep{hurley2011novelty}.   


Our following second main result shows that this PoA upper bound is tight, up to negligible constants. 


\begin{theorem}\label{thm:PoA_lowerbound}
  Given any $0\leq\beta\leq 1 $, $n > 2$ and any $1\leq K\leq \min\{n-1, e^{\frac{1}{5\beta}}\}$, there exists a \game{} instance $\G(\{\S_i\}_{i=1}^n, \X, \sigma, \beta, K)$ such that
    \begin{equation}
        PoA(\G)> \frac{n-1}{n}+\frac{1}{1+5\beta\log K}.
    \end{equation}
\end{theorem}
This theorem also implies that the argument we employed for Theorem \ref{thm:PoA}, which is based on the \emph{smoothness} proof developed by \citet{roughgarden2015intrinsic}, yields a tight PoA bound for our proposed game. The tightness of the smoothness argument is itself an intriguing research question. Only three classes of games are known to enjoy a tight PoA bound derived from the smoothness argument: congestion games with affine cost \citep{christodoulou2005price}, second price auctions \citep{christodoulou2008bayesian}, and the valid utility game \citep{vetta2002nash},  which are all fundamental classes of games. Theorem \ref{thm:PoA_lowerbound} suggests that our \game{} subscribes to this list. The proof of Theorem \ref{thm:PoA_lowerbound} is to explicitly construct a game instance which provably yields the stated PoA lower bound (see Appendix \ref{app:poa_lower_proof}).



\subsection{Implications of the PoA Bounds} \label{sec:imp_poa}

We have discussed some direct implications of  Theorem \ref{thm:PoA}. Now we develop new results which are either derived from or can be compared to Theorem \ref{thm:PoA} and \ref{thm:PoA_lowerbound}. They will reveal additional insights from our main theoretical results.

\vspace{2mm}
\noindent 
\textbf{Welfare implications to \emph{learning} content creators. }  
The PoA bounds presented in Theorem \ref{thm:PoA} and \ref{thm:PoA_lowerbound} are based on the assumption that creators are aware of the game parameters and play some CCEs of the game. While CCE is a reasonable equilibrium concept, one potential critique is that to find the CCE, it is assumed that each creator has knowledge about the system parameters (e.g., all other creators' strategies and the $\sigma$ function), which can be unrealistic. Fortunately, in real-world scenarios where creators utilize no-regret algorithms to play a repeated \game{} with bandit feedback, we can still establish a slightly worse PoA upper bound leveraging the fact that the average strategy history of no-regret players converges to a CCE, as shown in the following Corollary \ref{coro:noregret_CCE}.

\begin{corollary}\label{coro:noregret_CCE}[\textbf{Dynamic Version of Theorem \ref{thm:PoA}}]
Suppose each player in a repeated \game{}  $\G(\{\S_i\}_{i=1}^n, \X, \sigma, \beta, K)$ independently executes some no-regret learning algorithm, with worst   regret $R(T) = \max_i R_i(T)$ as defined in \eqref{def:regret}. Then we have
\begin{equation}\label{eq:PoA_noregret}
\frac{\max_{\s\in\S}W(\s)}{\frac{1}{T}\sum_{t=1}^T \EE_{\s \sim \al^t}[W(\s)]}
< 1+\Big(1+\frac{n}{\beta\log K}\cdot \frac{R(T)}{T}\Big)\cdot \frac{1}{c(\beta,K)}, 
\end{equation}
where $\al^t$ denotes the joint-strategy distribution at step $t$ and $c(\beta,K)$ is the constant defined in  \eqref{eq:cbetaK_def}.
\end{corollary}
In other words, the average welfare across all rounds $\frac{1}{T}\sum_{t=1}^T \EE_{\s \sim \al^t}[W(\s)]$ is close to the maximum possible welfare $\max_{\s\in\S}W(\s)$, up to a constant factor. The quantity in the LHS of \eqref{eq:PoA_noregret} is also known as the ``price of total anarchy"  \cite{blum2008regret}. It is a substitute for PoA when we want to characterize the welfare of an outcome from repeated play which does not necessarily fall into any equilibrium concept. The proof of Corollary \ref{coro:noregret_CCE} is presented in Appendix \ref{app:poa_proof_regret}. Because $R(T)/T \rightarrow 0$ as $T \rightarrow \infty$ for any no-regret algorithm (most no-regret algorithms have $R(T) = O(\sqrt{T})$), the RHS of \eqref{eq:PoA_noregret} is still strictly less than 2 for any fixed constants $(n,\beta,K)$. 

Theorem \ref{thm:PoA} relies on two crucial platform features: 1. the player's utility in \eqref{eq:player_utility_original} is defined as the total user engagement that accounts for the user utility $\sigma(\s_i, \x_j) + \varepsilon_j$, as opposed to just ``user exposure'' (i.e., the expected total number of matches); 2. the platform uses the top-$K$ recommendation policy. Next, we illustrate the insights revealed from Theorem \ref{thm:PoA} with respect to these two key features.

\vspace{2mm}
\noindent
\textbf{The importance of rewarding user engagement rather than solely exposure. } 
 A key reason for the nice PoA guarantee in our \game{} is each player $i$'s  utility is chosen as the user engagement $\sum_j \EE[\sigma(\s_i, \x_j)+\varepsilon_i|\x_j \shortrightarrow \s_i]   \Pr(\x_j \shortrightarrow \s_i)$  in \eqref{eq:player_utility_original}, while not the following user exposure metric: 
 \begin{equation}\label{eq:player_utility_new}
     \text{User exposure for player }i: \quad  \sum_{j=1}^m  \Pr(\x_j \shortrightarrow \s_i).  
 \end{equation}   

Our next result shows that incentivizing creators to maximize user exposure can lead to significantly worse welfare.

\begin{proposition}\label{prop:unbound_user_welfare}
  Let $\tilde{\G}$ denote the variant of the \game{} $\G =(\{\S_i\}_{i=1}^n, \{\x_j\}_{j=1}^m, \sigma, \beta, K)$   by substituting player utility function in  \eqref{eq:player_utility_original} by the above \emph{user exposure} in \eqref{eq:player_utility_new}.  Then for any $K\geq 1, 0\leq\beta\leq \min\{0.14, \frac{1}{5\log K} \}$, there exist $\G$ and $\tilde{\G}$ such that 
    \begin{equation}\label{eq:poa_geq2}
        PoA(\tilde{\G})>2> PoA(\G).
    \end{equation}
     Moreover, when $K=1$ or $\beta$ approaches $0$, PoA$(\tilde{\G})$ can be arbitrarily large.
\end{proposition}

In stark contrast to Theorem \ref{thm:PoA} guaranteeing PoA$(\G)< 2$, Proposition \ref{prop:unbound_user_welfare} implies the deterioration of user welfare when content creators are incentivized to compete for the expected exposure of their content.  However, we find in practice both metrics are used: for example, user engagement has been used more often as a reward metric for established creators, whereas user exposure is used more for new creators \cite{meta_experiment,savy2019}. Our result serves as a theoretical defense for rewarding creators by user engagement if the system aims to improve overall welfare of the recommendations. 

To prove Proposition \ref{prop:unbound_user_welfare}, we construct a game instance in which the user welfare at NE is arbitrarily close to zero. 
Our construction also reveals interesting insights about situations where user welfare can be very bad. Hence, we briefly explain our construction here and leave our formal arguments in Appendix \ref{app:unbound_user_welfare}. 
Our constructed game has two groups of users: one \emph{dispersed} group that is fine with any content but is never very happy with it (i.e., a low relevant score for all content) and one \emph{focused} group who looks for a specific type of high-quality content (a high relevance score on such content); but only a small group of specialized creators can produce such high-quality content. However, if players are incentivized to compete for exposure, 
even creators from the small group tend to produce low-quality content that appeals to the dispersed group rather than high-quality content that benefits the focused group. This, in the worse case, can lead to arbitrarily worse welfare for the platform. 


\vspace{2mm}
\noindent
\textbf{The welfare efficiency of top-$K$ recommendation policy.  } 
One may wonder whether the top-$K$ recommendation is indeed a good policy for securing the platform's welfare, i.e., is it possible that other recommendation policies (e.g., a probabilistic policy based on Plackett-Luce model \citep{plackett1975analysis,luce1959individual}) may even lead to better equilibrium outcomes? 
Our following analysis, as a corollary of Theorem \ref{thm:PoA}, shows that the answer is to some extent \emph{no} since \emph{any} recommendation policy cannot be better than the top-$K$ rule by more than a tiny fraction of the theoretical optimality. We believe this finding also serves as a theoretical justification for the wide adoption of the top-$K$ principle in practice. 


\begin{corollary}\label{prop:welfare_loss}
Consider an arbitrary recommendation policy providing at most $K$ recommendations, which induces a different \game{} $\G'$. Let CCE$(\G')$ denote the corresponding CCE set of $\G'$ and $W(\G')=\min_{\al\in CCE(\G')}\EE_{\s\sim \al}[W(\s)]$ be its worst-case CCE welfare. Then we have  
    \begin{equation}\label{eq:welfare_loss_fraction}
        W(\G') \leq W(\G) + W^*_K / \Big(1+\frac{K\log (K+b)}{K+b}\Big),
    \end{equation}
    where $W^*_K$ is the best possible social welfare achieved via any centralized recommendation policy with $K$ slots.   

\end{corollary}

As indicated by \eqref{eq:welfare_loss_fraction}, the fraction of the loss of welfare is approximately $O(\frac{1}{\log K})$ as $\frac{K}{K+b}\sim O(1)$ when $K$ is large. The proof is straightforward based on Theorem \ref{thm:PoA} and can be found in Appendix \ref{app:welfare_loss}.

\subsection{Proof Highlights of Theorem \ref{thm:PoA}}\label{subsec:proof-sketch}

Our first step is to derive clean characterizations for the game primitives by utilizing properties of Gumbel distribution. The form of the user utility $\pi_j$ and welfare $W$ are corollaries of RU models \citep{baltas2001random}, however, the closed-form of the creator utility $u_i$ is a new property we derive. The results are summarized in the following Lemma \ref{lm:closed_form_utility}, and the detailed proof is deferred to Appendix \ref{app:close_utility}. 

\begin{lemma}\label{lm:closed_form_utility}
Given $\{ \varepsilon_i \}$ are drawn i.i.d.  from zero-mean Gumbel$(-\beta \gamma, \beta)$, the utility and welfare functions defined in \eqref{eq:user_utility_original}, \eqref{eq:player_utility_original} and \eqref{eq:welfare_original} have the following closed forms
\begin{equation}\label{eq:user_utility}
        \pi_j(\s)= \beta\log \Big[\sum_{\s_k\in\T_j(\s;K)} \exp{(\beta^{-1}\sigma(\s_k, \x_j))}\Big],
\end{equation}
\begin{equation}\label{eq:player_utility}
  u_i(\s)= \sum_{j=1}^m \pi_j(\s)\frac{\II[\s_i\in \T_j(\s;K)]\exp(\beta^{-1}\sigma(\s_i, \x_j))}{\sum_{\s_k\in \T_j(\s;K)} \exp(\beta^{-1}\sigma(\s_k, \x_j))},
\end{equation}
\begin{equation}\label{eq:welfare}
\!\!\!\!        W(\s)= \beta\sum_{j=1}^m \log \Big[\sum_{\s_k\in\T_j(\s;K)} \exp{(\beta^{-1}\sigma(\s_k, \x_j))}\Big].
\end{equation}
\end{lemma}

The main proof of Theorem \ref{thm:PoA} is based on a smoothness argument framework developed in the seminal work by \citet{roughgarden2015intrinsic}. For any strategy profile $\s$,  $W(\s) = \sum_i u_i(\s_i, \s_{-i}) $ is its  total  welfare function. A game is   $(\lambda,\mu)$\emph{-smooth} if  $ \lambda W(\s') - \mu W(\s) \leq \sum_i u_i(\s'_i, \s_{-i}) $ for any $(\s, \s') \in \S$. \citet{roughgarden2015intrinsic} observes that the PoA of any $(\lambda,\mu)$-smooth game  can be upper bounded by $\frac{1 + \mu}{\lambda}$. After plugging in the expression of $W(\s)$, the $(\lambda,\mu)$-smoothness condition can be re-written as
$$ 
    \sum_i [ \lambda u_i(\s'_i, \s'_{-i}) - u_i(\s'_i, \s_{-i})  ] \leq \sum_i  \mu u_i(\s) . 
$$
Intuitively, the smoothness parameters bound how much \emph{externality} other players' actions  (i.e., $\s'_{-i}$ or $\s_{-i}$) impose on any player $i$'s utility. Moreover, the tighter this bound is, the smoother the game is and the smaller the PoA is. To gain some intuition and also as a sanity check, consider the extreme situation in which each player's utility is not affected by other players' actions at all (i.e., the \emph{no externality} situation), we have $\lambda = 1$ and $\mu = 0$ implying PoA=1. That is, if any player's utility is not affected by others, then self-interested utility-maximizing behaviors also maximize social welfare, which is a straightforward observation. Certainly, we cannot hope for such a nice property to hold in general, but fortunately, many well-known games have been shown to be smooth. For example, second-price auctions are   $(1,1)$-smooth as shown by \citet{christodoulou2008bayesian},   congestion games are $(\frac{5}{3},\frac{1}{3})$-smooth as shown by \citet{roughgarden2015intrinsic}, and all-pay auctions are $(1/2,1)$-smooth as shown by \citet{roughgarden2017price}. 

Hence, the key challenge in proving Theorem \ref{thm:PoA} is to pin down the tightest possible $(\lambda, \mu)$ parameters for our \game{}. This boils down to a fundamental question in top-$K$ RS -- i.e., \emph{to what extent does the existence of other competing content creators affect a creator's utility}? To answer this question, we discover multiple interesting properties of the welfare and creator utility functions formulated as follows. Besides proving our main result in Theorem \ref{thm:PoA}, we believe these properties are also of interest for us to understand recommender systems. 

Our second Lemma \ref{lm:submodular_V_} demonstrates the \emph{submodularity} of $W(\s)$. That is, the marginal gain of welfare from adding a new player decreases as the total number of creators increases. Lemma \ref{lm:utility_property_} further relates this marginal welfare increase with the added player's own utility. It shows that the increased welfare after introducing a new player $i$ with strategy $\s_i$ is at most   $i$'s utility  under  $\s_i$, multiplied by a shrinkage factor $c^{-1}(\beta, K) \in (0, 1]$. These two lemmas together allow us to prove that the \game{} is $(c^{-1}(\beta, K), c^{-1}(\beta, K))$-smooth, yielding Theorem \ref{thm:PoA}. Detailed proofs are presented in Appendix \ref{app:lemma23_proof}. 

\begin{lemma}\label{lm:submodular_V_}[Submodularity of Welfare]
For any $\s=(\s_1,\cdots,\s_n)\in\S$, let $S=\{\s_1,\cdots,\s_n\}$. Then the social welfare function defined in Eq \eqref{eq:welfare_original} is submodular as a set function, i.e., for any $S, \s_x, \s_y$ it holds that
\begin{equation*}
    W(S\cup \{\s_x\})-W(S) \geq W(S\cup \{\s_x,\s_y\})-W(S\cup \{\s_y\}).
\end{equation*}
\end{lemma}

\begin{lemma}\label{lm:utility_property_} [Smoothness of Welfare]
For any $\s=(\s_1,\cdots,\s_n)\in\S$, $i\in[n]$ and $c(\beta, K)$ defined in Eq \eqref{eq:cbetaK_def}, player-$i$'s utility function $u_i(\s)$ defined in Eq \eqref{eq:player_utility_original} satisfies
\begin{equation*}
    W(\s)-W(\s_{-i}) \leq c^{-1}(\beta, K)\cdot u_i(\s_i;\s_{-i}).
\end{equation*}
\end{lemma}

\section{Experiments}

To confirm our theoretical findings and also to empirically measure the social welfare induced by creators' competition, we conduct simulations on game instances \\ $\G(\{\S_i\}_{i=1}^n, \X, \sigma, \beta, K)$ constructed from two synthetic datasets and the MovieLens-1m dataset \citep{harper2015movielens}. Before presenting our results, we provide a detailed overview of the simulation environment, including the characteristics of the datasets utilized and the metrics employed for evaluation.

\subsection{Specification of Datasets}
\noindent
{\bf $\bullet$ Synthetic Dataset-1}
Dataset-1 simulates the situation where content creators compete over an unbalanced user interest distribution. We construct $n$ user clusters with the largest cluster containing half of the population, and let each strategy from a creator's action set generate content that only appeals to a specific user group. 

Specifically, the user population is given by disjoint clusters $\X=\cup_{i=1}^n \X_i$ such that $|\X_1|=\frac{m}{2}$, and the sizes of smaller clusters $|\X_l|$ are sampled uniformly at random such that $\sum_{l=2}^n |\X_l|=\frac{m}{2}$. Players share the same action set $\S_i=\{\s_1,\cdots,\s_n\}$, and the $\sigma$ function satisfies that for any $i\in [n]$,

\begin{equation}  
\sigma(\s_i, \x)=\left\{  
             \begin{array}{lr}  
             1, \text{~~if~~} \x\in \X_i,&   \\  
             0, \text{~~otherwise.}&    
             \end{array}  
\right.  
\end{equation} 

Dataset-1 depends on the randomness of the partition $\cup_{i=1}^n \X_i$.

\vspace{2mm}
\noindent
{\bf $\bullet$ Synthetic Dataset-2}
Dataset-2 simulates the situation where content creators can either ``chase the trend" by generating mediocre content or cater to a specific user interest group with high-quality content. Similar to the construction of dataset-1, we let the user population comprise of $n$ clusters and allow each player to take actions targeting at any specific user group. But, in addition, we also allow each player to take a ``safe" action $\s_0$ by producing some popular content that can satisfy all users to a certain extent $\delta$.

Specifically, the user population is also given by disjoint clusters $\X=\cup_{i=1}^n \X_i$, the sizes of all clusters $|\X_l|$ are sampled uniformly at random such that $\sum_{l=i}^n |\X_l|=m$. Players share the same action set $\S_i=\{\s_0, \s_1,\cdots,\s_n\}$, and the $\sigma$ function satisfies that for any $i\in [n]$,

\begin{equation}  
\sigma(\s_i, \x)=\left\{  
             \begin{array}{lr}  
             1, \text{~~if~~} \x\in \X_i, i\geq 1&   \\  
             \delta, \text{~~if~~} i=0&   \\  
             0, \text{~~otherwise.}&    
             \end{array}  
\right.  
\end{equation} 

Dataset-2 depends on the randomness of the partition $\cup_{i=1}^n \X_i$ and a parameter $\delta\in[0,1]$.

\vspace{2mm}
\noindent
{\bf $\bullet$  The Dataset Generated from MovieLens-1m}
We use deep matrix factorization \citep{fan2018matrix} to train
user and movie embeddings targeted at movie ratings from 1 to 5. The total number of users $m=6040$, the number of movies $k=3883$, and the embedding dimension is set to $d=32$. To validate the quality of the trained representation, we first performed a 5-fold cross-validation and obtain an averaged RMSE $=0.883$ on the test sets, then train the user/item embeddings with the complete dataset. The resulting user embeddings $\X=\{\x_j\}_{j\in [m]}$ are used as the user population. To construct each player-$i$'s action set $\S_i$, we randomly sample $500$ vectors from the trained movie embedding set $\M$ ($|\M|=3883$) independently. To normalize the relevance score to $[0, 1]$, we let $\sigma(\s, \x)=1$ when the predicted rating of movie $\s$ to user $\x$ is at least 4, i.e., $\sigma(\s, \x)=\II[\langle \s, \x\rangle\geq 4]$. 

\subsection{Evaluation Metrics}

We use both PoA and PotA in our experiments. The evaluation of PoA requires solving two optimization problems, which are both intractable in general due to the non-concavity of $W(\cdot)$ and the undetermined structure of CCE$(\G)$. As a result, we use simulated annealing to approach $\max_{\s\in\S}W(\s)$ when the exact computation is intractable. To compute $\min_{\al\in CCE(\G)}\EE_{\s\sim \al}[W(\s)]$, we compute its exact solution by solving a linear program with $k^n$ variables and $kn$ constraints \citep{papadimitriou2005computing} for small $n$ and a moderate size of action set $k$. To deal with larger problems, we let each player run Exp3 \cite{auer2002nonstochastic} over $T=5000$ rounds and compute the price of total anarchy PotA$(\G)=\frac{\max_{\s\in\S}W(\s)}{\frac{1}{T}\sum_{t=1}^T \EE_{\s \sim \al^t}[W(\s)]}$. More details are disclosed in Appendix \ref{app:setup}.

\subsection{Results}
\noindent 
{\bf Empirical PoA from simulations.} We first demonstrate the empirical welfare under different game parameter $(n, K, \beta)$ for dataset-1. We fix $\beta \in \{0.1,0.5\}$ and report PoA and PotA under varying $n$ and $K$. Results reported in Table \ref{tb:exp1_small_n}, \ref{tb:exp1_large_n}, \ref{tb:exp1_small_n_beta0.5} and \ref{tb:exp1_large_n_beta0.5}. We observe that for any fixed $n$, both PoA and PotA decrease w.r.t. $K$ and $\beta$, as revealed in Theorem \ref{thm:PoA}. Furthermore, under fixed $(\beta, K)$, PoA approaches its theoretical upper bound as $n$ increases. However, PotA follows this trend for values of $n$ less than 15, but begins to decrease as $n$ increases further. This discrepancy can be attributed to the fact that for larger values of $n$ (i.e., in Table \ref{tb:exp1_large_n}), the approximated optimal welfare becomes less accurate and as such, the PotA tends to be underestimated. 
\begin{table}[t]
\centering
 \caption{PoA under $\beta=0.1$. Results reflect the worst cases obtained from 10 independently sampled game instances. }
 \label{tb:exp1_small_n}
\begin{tabular}{ |c|c|c|c|c|c|  }
 \hline
  \diagbox{$K$}{$n$} & * & 2 & 3 & 4 & 5 \\
 \hline
 $1$ &2.00 & 1.33 & 1.54 & 1.66 & 1.72   \\
 $2$ &1.93 & 1.28 & 1.46  & 1.56  &  1.60  \\
 $3$ &1.89 & \/ & 1.42 &  1.47   & 1.51 \\
 $4$ &1.86 & \/& \/    & 1.43    &  1.42\\
  $5$ &1.84 & \/& \/    & \/    &  1.42\\
 \hline
\end{tabular}
\\\text{$*$ denotes the theoretical upper bound.}
 \end{table}

\begin{table}[t]
\centering
\caption{PotA under $\beta=0.1$. Results reflect the worst cases obtained from 10 independently sampled game instances.}
 \begin{tabular}{ |c|c|c|c|c|c|c|  }
\hline
  \diagbox{$K$}{$n$} & * & 5 & 10 & 15 & 20 & 40 \\
 \hline
 $1$ &2.00 & 1.59 & 1.59 & 1.60 & 1.50 & 1.38   \\
 $3$ &1.89 & 1.37 & 1.39 & 1.42 & 1.41  &  1.32  \\
 $5$ &1.84 & 1.35 & 1.34 & 1.33 & 1.36   & 1.31 \\
 $7$ &1.80 & \/& 1.30    & 1.31 & 1.30    &  1.29\\
 \hline
\end{tabular}
\\\text{$*$ denotes the theoretical upper bound.}
 \label{tb:exp1_large_n}
\end{table}

\begin{table}[t]
\centering
 \caption{PoA under $\beta=0.5$. Results reflect the worst cases obtained from 10 independently sampled game instances. }
 \label{tb:exp1_small_n_beta0.5}
\begin{tabular}{ |c|c|c|c|c|c|  }
 \hline
  \diagbox{$K$}{$n$} & * & 2 & 3 & 4 & 5 \\
 \hline
 $1$ &2.00 & 1.33 & 1.54 & 1.66 & 1.72   \\
 $2$ &1.77 & 1.11 & 1.24  & 1.32  &  1.34  \\
 $3$ &1.65 & \/ & 1.08 &  1.13  & 1.18 \\
 $4$ &1.57 & \/& \/    & 1.05    &  1.08\\
  $5$ &1.52 & \/& \/    & \/    &  1.02\\
 \hline
\end{tabular}
\\\text{$*$ denotes the theoretical upper bound.}
 \end{table}

\begin{table}[t]
\centering
\caption{PotA under $\beta=0.5$. Results reflect the worst cases obtained from 10 independently sampled game instances.}
 \begin{tabular}{ |c|c|c|c|c|c|c|  }
\hline
  \diagbox{$K$}{$n$} & * & 5 & 10 & 15 & 20 & 40 \\
 \hline
 $1$ &2.00 & 1.59 & 1.59 & 1.60 & 1.50  & 1.38   \\
 $3$ &1.65 & 1.13 & 1.20 & 1.21 & 1.22  & 1.20  \\
 $5$ &1.52 & 1.03 & 1.10 & 1.12 & 1.14  & 1.14 \\
 $7$ &1.45 & \/& 1.05    & 1.08 & 1.09  &  1.11\\
 \hline
\end{tabular}
\\\text{$*$ denotes the theoretical upper bound.}
 \label{tb:exp1_large_n_beta0.5}
\end{table}

\noindent
{\bf Comparison between user engagement/exposure metrics.}
Next we investigate the consequence of utilizing two different incentive metrics, namely \emph{user engagement} vs., \emph{user exposure}. However, Dataset-1 is no longer a perfect benchmark for this purpose, as the utility functions derived under a simple binary valued $\sigma(\cdot,\cdot)$ are almost indistinguishable under these two metrics. To this end, we use dataset-2, which has a more complex $\sigma(\cdot,\cdot)$ function that models the situation in which creators could focus on chasing the trends other than paying attention to the content quality. 

 We fix $(\beta, K)=(0.1, 2)$ and report PotA under different $n$ and $\delta$. The results, shown in Figure \ref{fig:two_metric_k=2_beta=0.1}, demonstrate the advantage of using the user engagement metric, which consistently leads to a smaller PotA across different values of $n$ and $\delta$. 
For $n$ larger than $10$, PotA with user-exposure can exceed $2$ as revealed by Proposition \ref{prop:unbound_user_welfare}\footnote{Again, due to the approximation error in computing optimal $W$, the PotA could be underestimated as $n$ gets larger.  }. The performance gap between the two metrics is more distinct when $\delta$ gets smaller, which can be understood as when creators can produce popular content with lower effort, simply using exposure to reward creators can be catastrophic to the total user welfare. Results in Figure \ref{fig:two_metric_k=2_2} illustrate the trend of PotA under different $(\beta, K)$, which further demonstrate the superiority of utilizing the user engagement metric. It is also interesting to point out Figure \ref{fig:two_metric_k=2_2} suggests that the advantage of the engagement metric over the exposure metric diminishes as either $\beta$ or $K$ increases. This implies that the PotA under the exposure metric also decreases with respect to $K$ and $\beta$, although this claim remains unproven. Nonetheless, it constitutes an interesting question for further investigation.

\begin{figure}[t]
\begin{center}
\centerline{\includegraphics[width=\columnwidth]{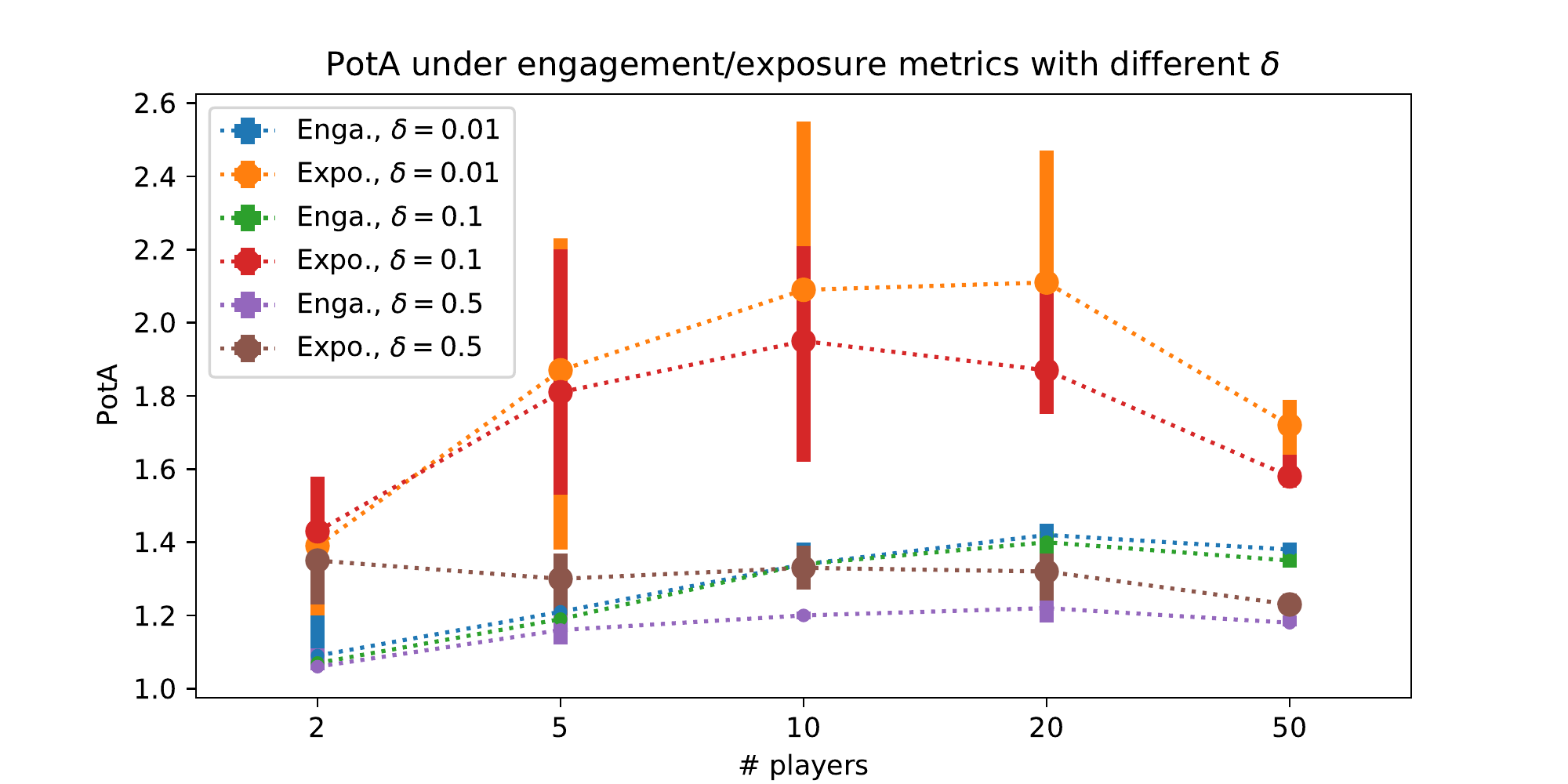}}
\vskip -0.1in
\caption{PotA under exposure/engagement metrics with $(\beta, K)=(0.1, 2)$. $\delta$ is the relevance score obtained from creators' ``safe" action. The error bars indicate the largest/smallest values from 10 independent trials and the dots correspond to the mean values.  }
\label{fig:two_metric_k=2_beta=0.1}
\end{center}
\vskip -0.4in
\end{figure}

\begin{figure}[t]
\includegraphics[width=0.5\columnwidth]{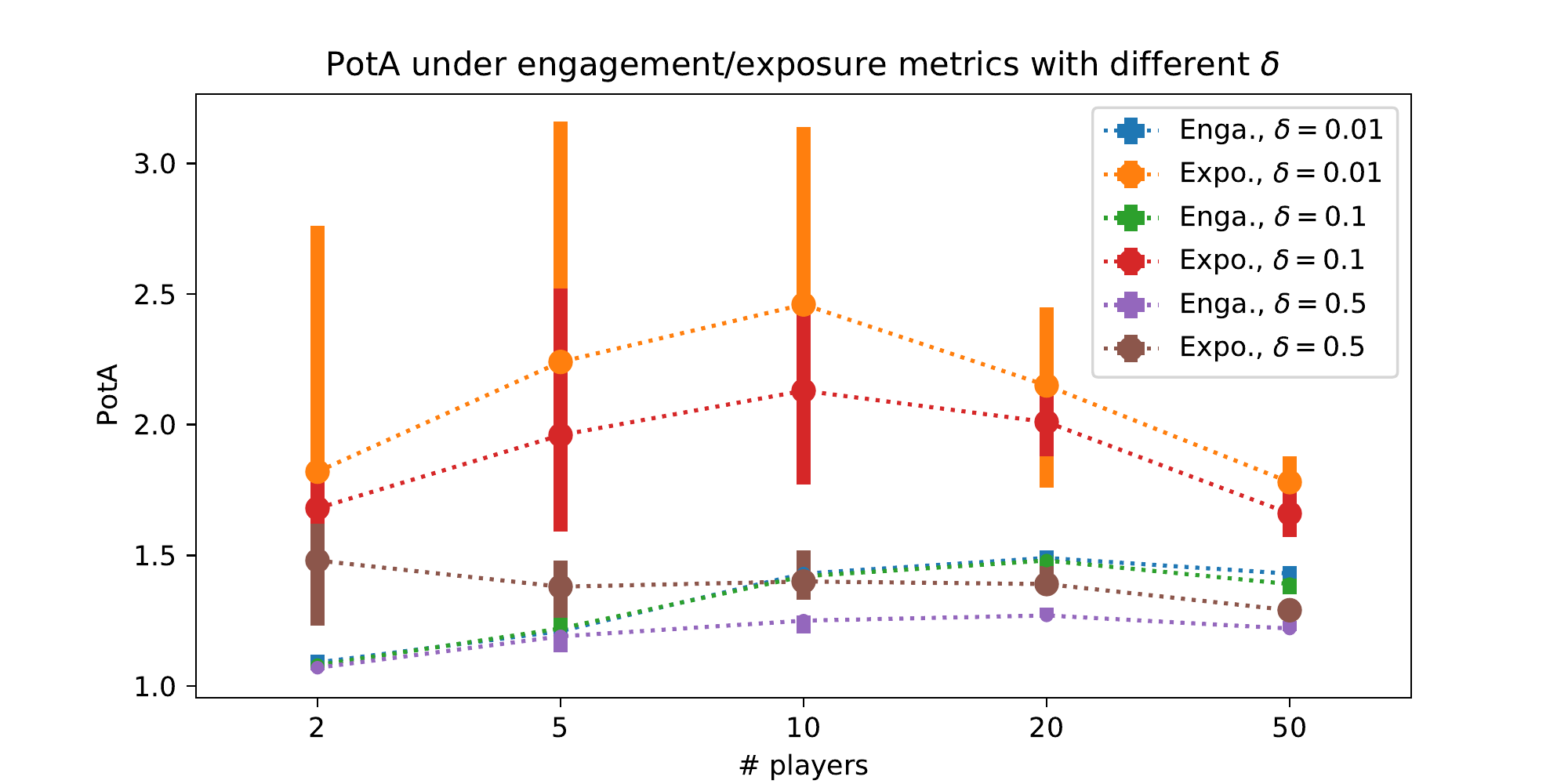}
\includegraphics[width=0.5\columnwidth]{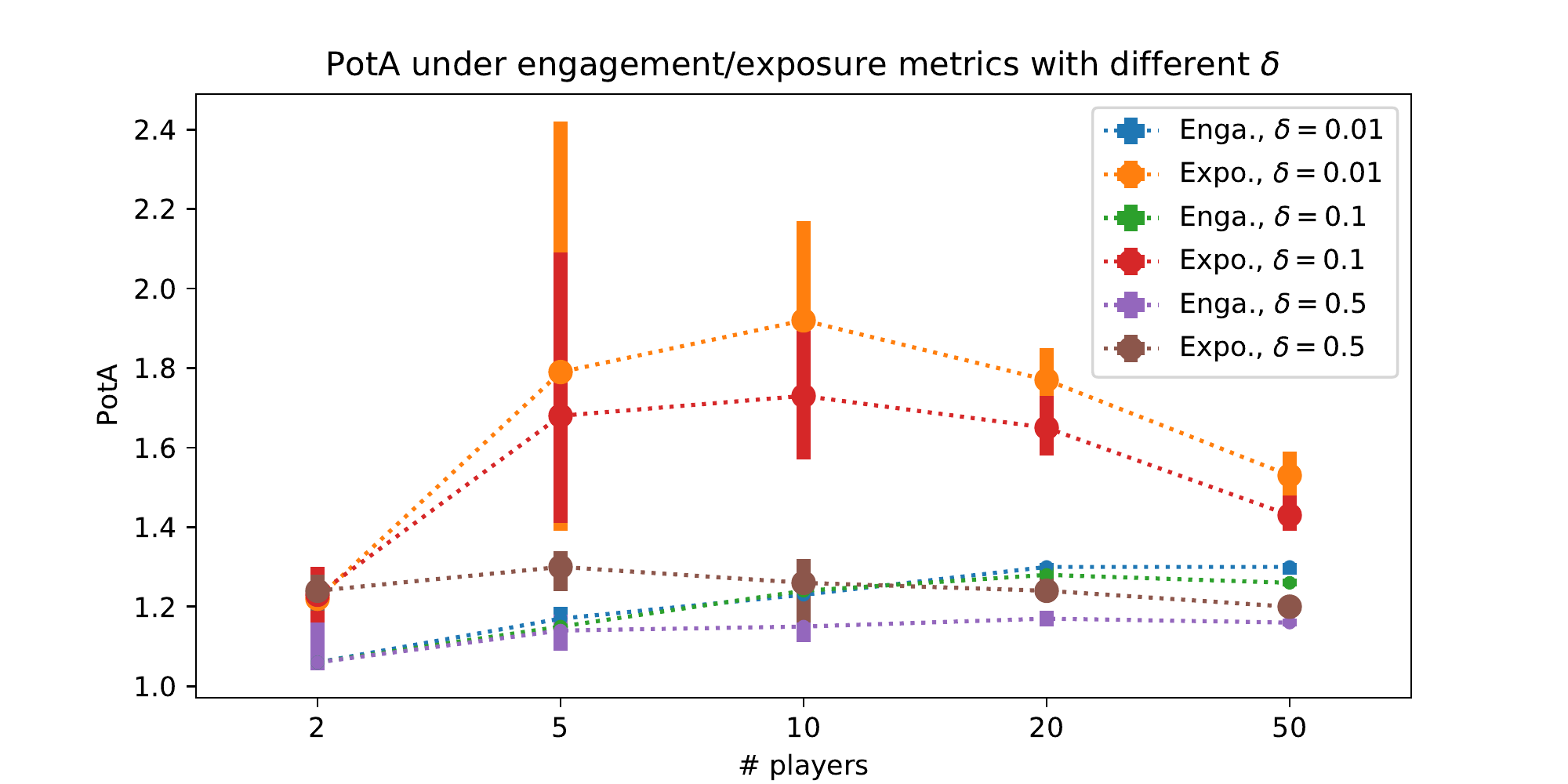}
\includegraphics[width=0.5\columnwidth]{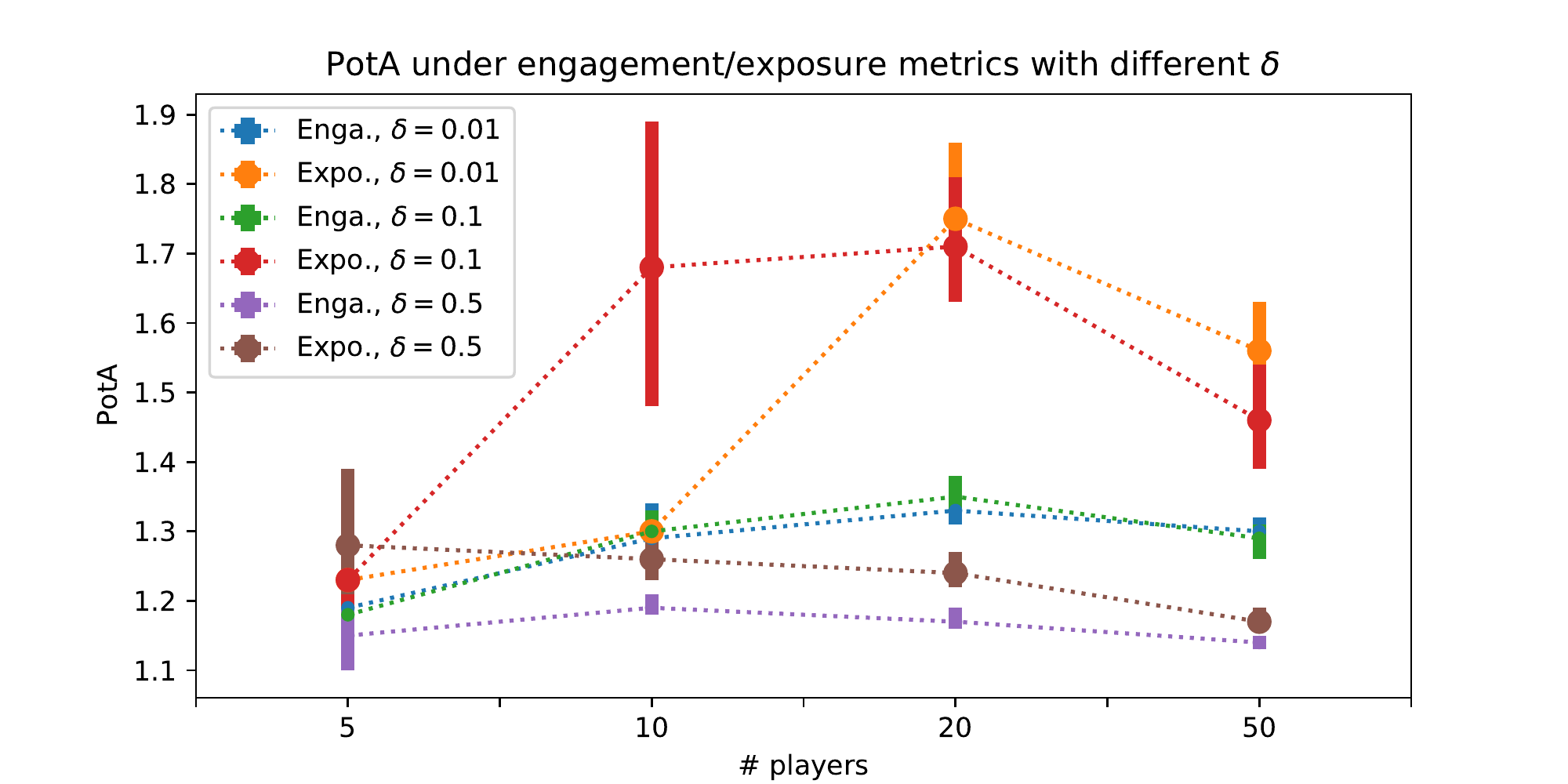}
\includegraphics[width=0.5\columnwidth]{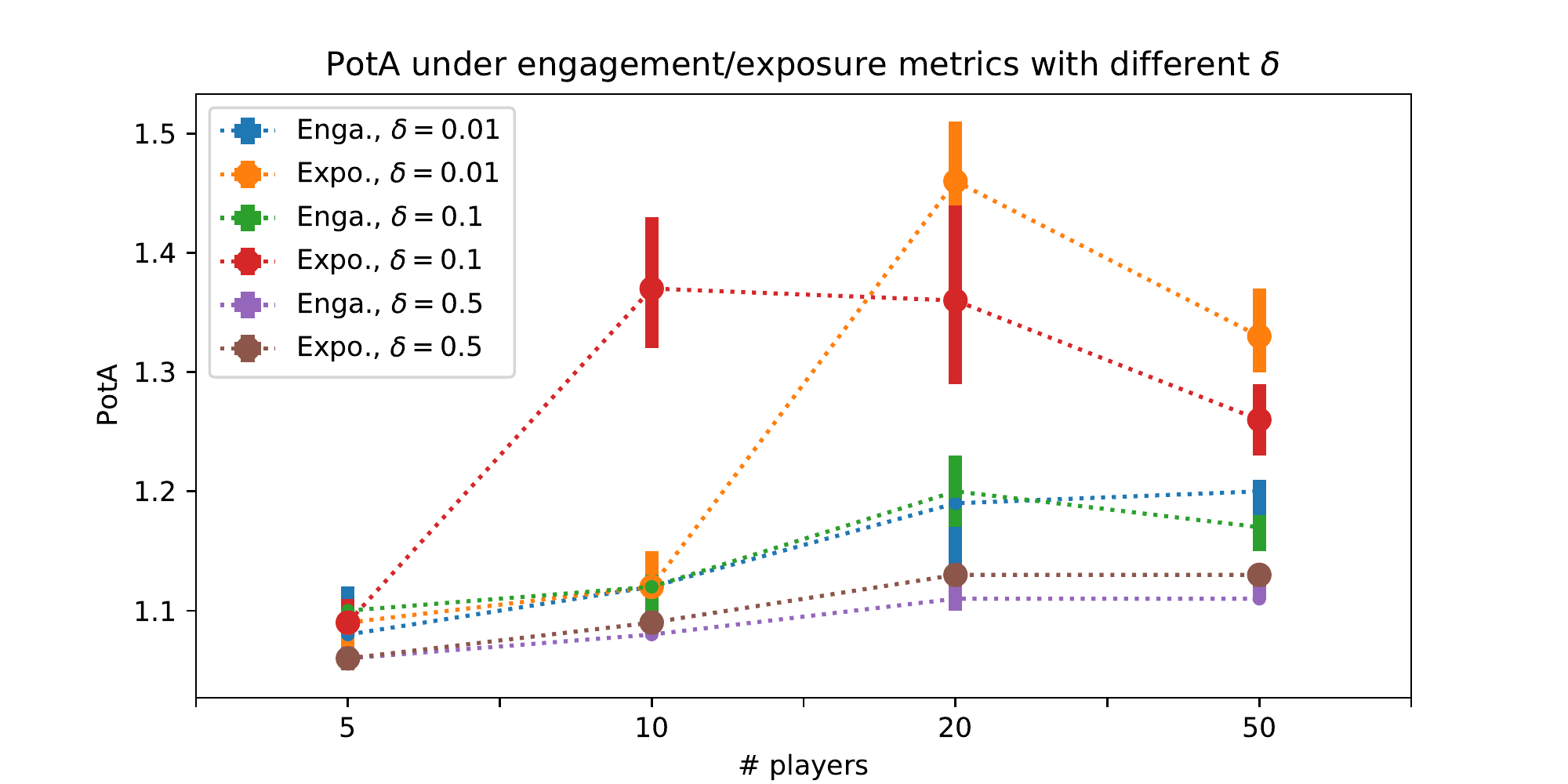}
\caption{PotA under user exposure and user engagement metrics under different $ (\beta, K)$. $\delta$ is the relevance score obtained from creators' ``safe" action. Upper-left: $(0.1, 1)$, Upper-right: $(0.3, 2)$, Lower-left: $(0.1, 5)$, Lower-right: $(0.3, 5)$. The error bars indicate the largest/smallest values from 10 independent trials and the dots correspond to the mean values.}
\label{fig:two_metric_k=2_2}
\end{figure}

\noindent
{\bf Social welfare under different levels of rationality.}
In this experiment, we aim to investigate the competition outcomes when players utilize online-learning algorithms with varying levels of rationality. To better simulate what happens in practice, we employed the dataset generated from MovieLens-1m \citep{harper2015movielens}. In our simulation, we model the scenario in which each player runs Exp3 under different exploration rates $\epsilon$ (i.e., with probability $\epsilon$, each player will take a random action in each round). We use this simulation setup to examine a practical situation, i.e., creators try to optimize their accumulated regret but with bounded rationality: since Exp3 is known to enjoy a sub-linear regret when $\epsilon \sim O(\sqrt{\frac{k\log k}{T}})$ \citep{auer2002nonstochastic}, it would be less rational for creators to set $\epsilon$ to be too large or too small as it would incur a larger regret $R(T)$. We fix $(\beta, K, T)=(0.1, 5, 1000)$ and report the averaged social welfare over $T$ rounds, i.e., $\bar{W}=\frac{1}{mT}\sum_{t=1}^T W(\s^{(t)})$ under different $n$ and $\epsilon$, as illustrated in Figure \ref{fig:movieLens_w_1}. Parallel results obtained with $(\beta, K)=\{(0.5,1),(0.5,5)\}$ are illustrated in \ref{fig:movieLens_w_2}. The plots in both Figure \ref{fig:movieLens_w_1} and \ref{fig:movieLens_w_2} indicate that the optimal exploration rate associated with the maximum welfare for $K=5$ is around $\epsilon=0.1$ across different values of $n$. When $\epsilon$ is set to be either too small or too large, the average welfare decreases, thereby confirming our claim in Corollary \ref{coro:noregret_CCE} that the welfare guarantee deteriorates as the accumulated regret of each player's learning algorithm increases. 
Additionally, we observed that the average welfare increases when more creators are involved, which is expected given that users will have a higher chance of receiving a satisfactory recommendation when there is a larger pool of content on the platform. Furthermore, when the number of players is sufficiently large ($n=100$), the welfare is fairly good even when players adopt nearly randomized strategies ($\epsilon=0.9$). 

Interestingly, the optimal exploration rate for $K=1$ is around $\epsilon=0.3$, which is higher than the optimal $\epsilon$ for $K=5$ as shown in the left panel of Figure \ref{fig:movieLens_w_2}. A potential explanation for this discrepancy is that when the RS provides a greater number of alternatives in the recommendation list, the randomness in users' decisions can better compensate for creators' exploration in finding more diversified strategies, leading to improved social welfare.

\begin{figure}[t]
\begin{center}
\centerline{\includegraphics[width=\columnwidth]{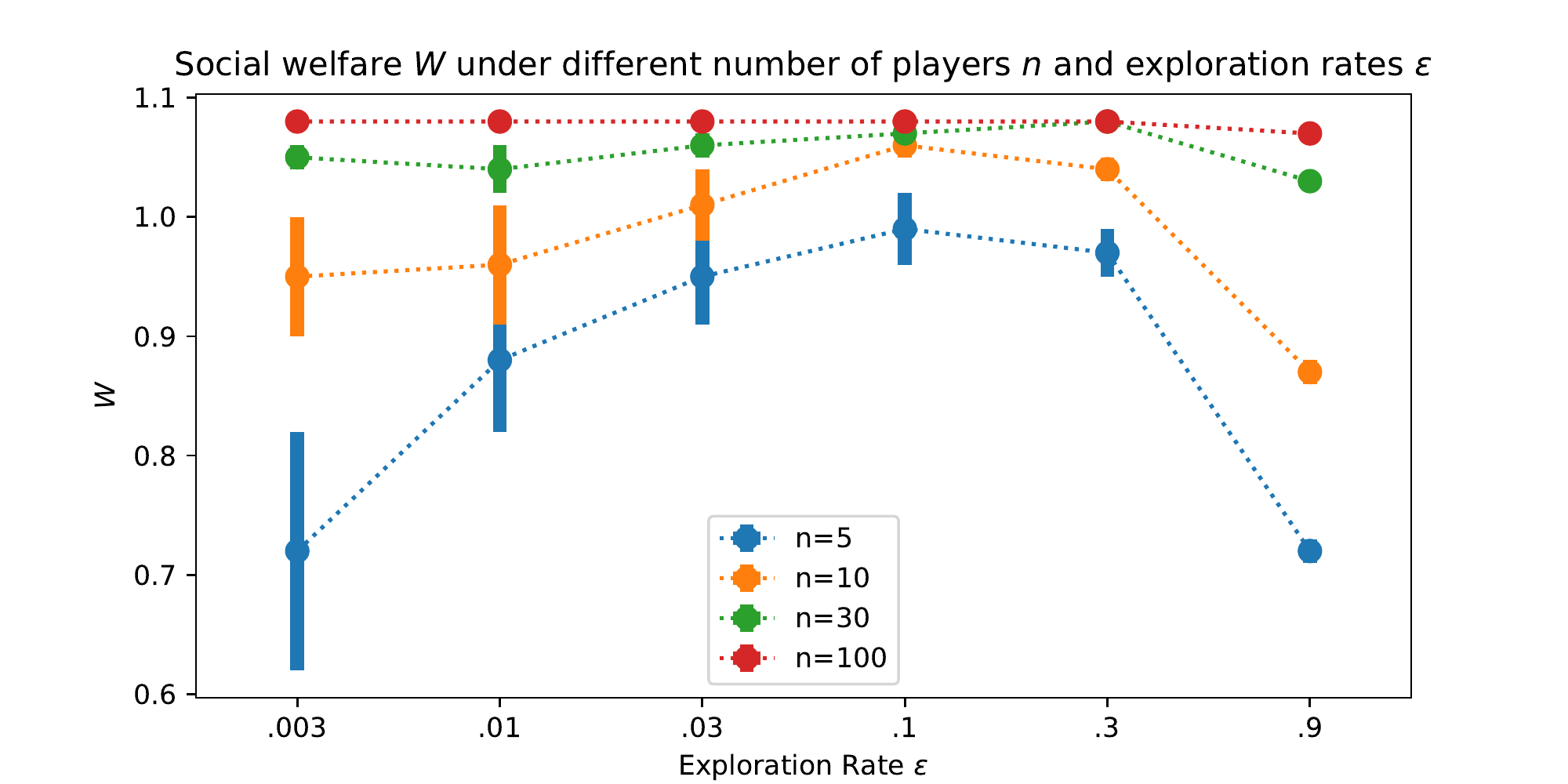}}
\vskip -0.1in
\caption{The averaged welfare $\bar{W}$ over $T=1000$ rounds under different exploration rate $\epsilon$ and number of players $n$. Results are averaged over 10 independent runs under $(\beta, K)=(0.1, 5)$.}
\label{fig:movieLens_w_1}
\end{center}
\end{figure}

\begin{figure}[t]
\includegraphics[width=0.5\columnwidth]{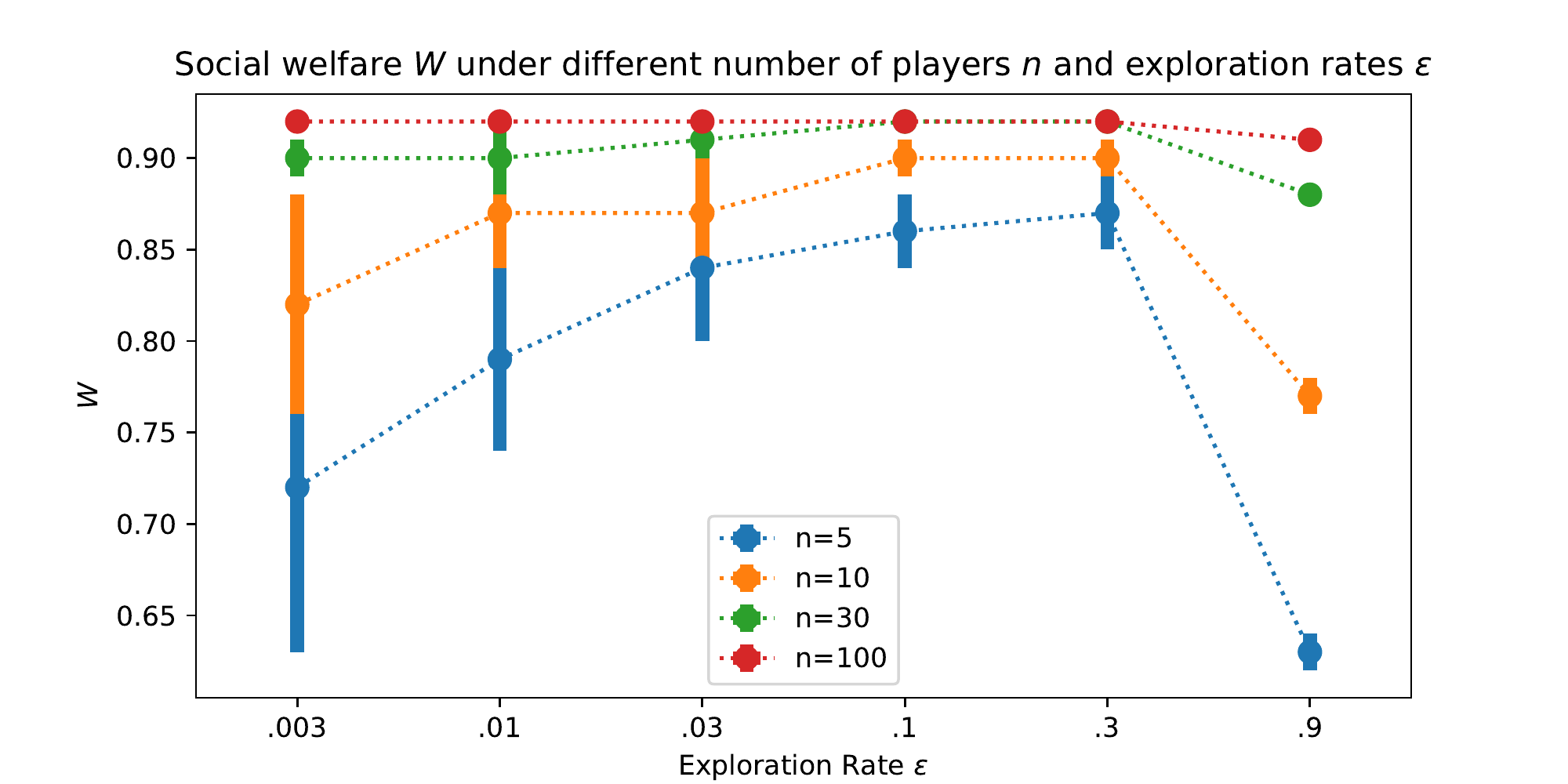}
\includegraphics[width=0.5\columnwidth]{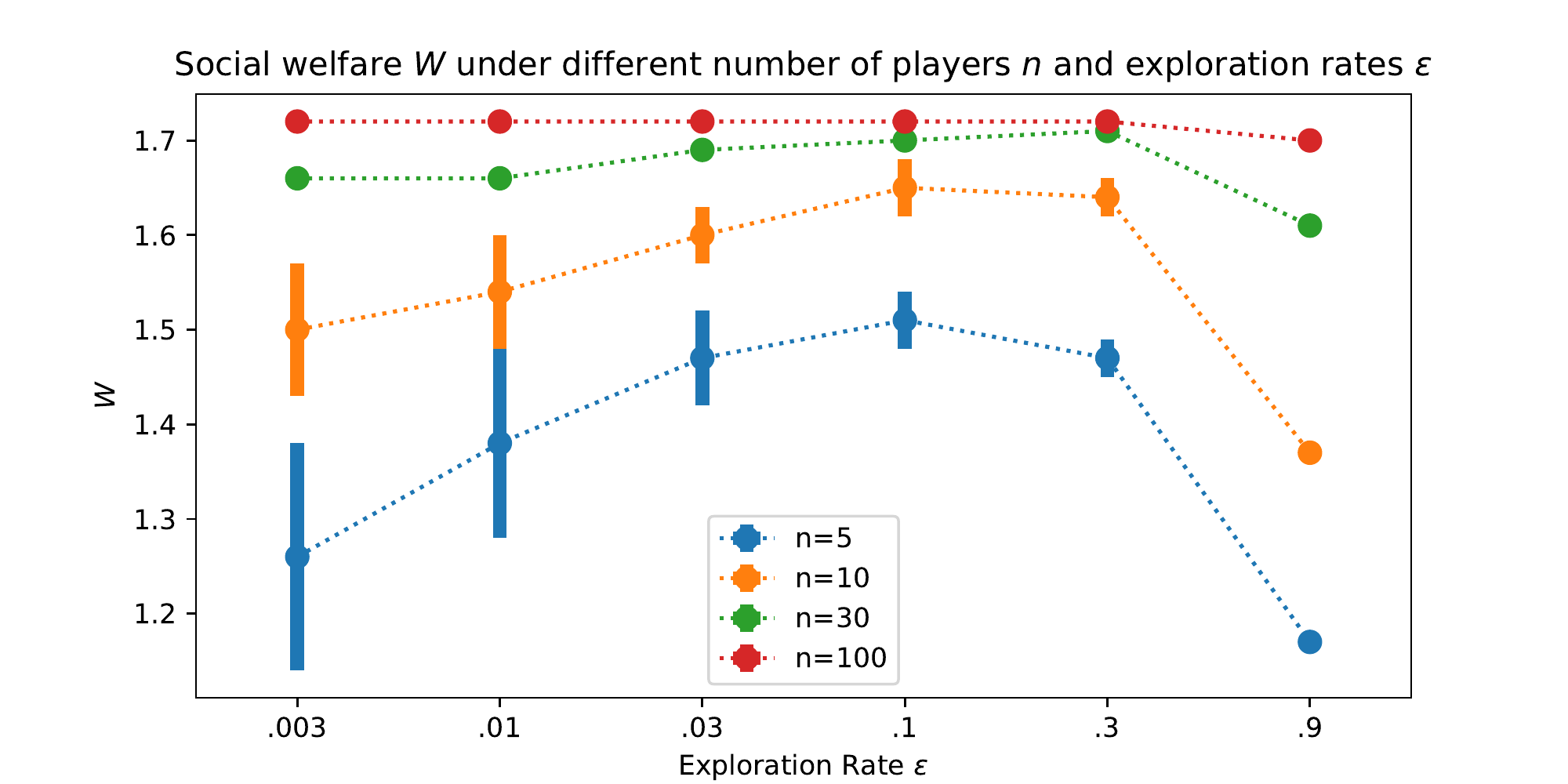}
\caption{The averaged welfare $\bar{W}$ over $T=1000$ rounds under different exploration rate $\epsilon$ and number of players $n$. Results are averaged over 10 independent simulations under $(\beta, K) \in \{(0.5, 1), (0.5,5)\}$.}
\label{fig:movieLens_w_2}
\end{figure}

\noindent{\bf The PotA for MovieLens dataset}
Under the same setting used in Figure \ref{fig:movieLens_w_1} and \ref{fig:movieLens_w_2}, where we computed the average welfare for different numbers of players with various exploration rates, we also present the corresponding PotA in Table \ref{tb:exp3_large_n_beta0.1} and \ref{tb:exp3_large_n_beta0.5}. As can be observed, when creators are able to utilize the optimal exploration rates that minimize regret, social welfare is unexpectedly high. Even with as few as five creators, the fraction of welfare loss is only around $10\%$ (PotA=1.10). As the number of creators increases, the PotA rapidly approaches the ideal value of 1.0. Although the PotA may be underestimated due to approximation accuracy in determining the optimal welfare, such results still convey a very optimistic message. 

\begin{table}[t]
\centering
\caption{PotA under $(\beta, K)=(0.1, 5)$ for the MovieLens dataset. Results are averaged over 10 independent simulations under different exploration factor $\epsilon$.}
 \begin{tabular}{ |c|c|c|c|c|c|c|  }
\hline
  \diagbox{$n$}{$\epsilon$} & .003 & .01 & .03 & .1 & .3 & .9 \\
 \hline
 $5$  & 1.51 & 1.32 & 1.18 & 1.10  & 1.12 & 1.51  \\
 $10$  & 1.15 & 1.14 & 1.08 & 1.03  & 1.05 & 1.25  \\
 $30$  & 1.06 & 1.05 & 1.04 & 1.02  & 1.03 & 1.07 \\
 $100$  & 1.02 & 1.02  & 1.02 & 1.02  &  1.02 & 1.03\\
 \hline
\end{tabular}
 \label{tb:exp3_large_n_beta0.1}
\end{table}

\begin{table}[t]
\centering
\caption{PotA under $(\beta, K)=(0.5, 5)$ for the MovieLens dataset. Results are averaged over 10 independent simulations under different exploration factor $\epsilon$.}
 \begin{tabular}{ |c|c|c|c|c|c|c|  }
\hline
  \diagbox{$n$}{$\epsilon$} & .003 & .01 & .03 & .1 & .3 & .9 \\
 \hline
 $5$  & 1.37 & 1.25 & 1.18 & 1.15  & 1.18 & 1.48  \\
 $10$  & 1.15 & 1.12 & 1.08 & 1.05  & 1.06 & 1.26  \\
 $30$  & 1.04 & 1.04 & 1.02 & 1.02  & 1.01 & 1.07 \\
 $100$  & 1.01 & 1.01  & 1.01 & 1.01  &  1.01 & 1.02\\
 \hline
\end{tabular}
 \label{tb:exp3_large_n_beta0.5}
\end{table}

\vspace{2mm}
\noindent{\bf The genre distribution resulting from the sequence of play}
In addition to evaluating the welfare metric, we also examined the content distribution induced by creators' sequence of play under different exploration rates, as shown in Figure \ref{fig:movieLens_genre_distribution}. The content distribution is generated by counting the frequency of genres (in MovieLens-1m dataset, each movie is associated with multiple genre tags) throughout a simulation trajectory. Specifically, for each simulation path with a fixed time horizon $T=1000$, we collect the complete set of content that all creators have produced and summarize their frequencies w.r.t. genres. We also compute the genre distributions for the optimal joint strategy that maximizes the social welfare under different $n$ as references.  

The leftmost panel in Figure \ref{fig:movieLens_genre_distribution} illustrates the genre distribution of the globally optimal content creation strategies under varying numbers of creators $n$. As $n$ increases, the optimal genre distribution becomes ``flatter", which is consistent with the expectation that new creators who want to maximize their utility should focus on niche interest groups to diversify the overall content distribution. The empirical genre distributions from simulations, shown in the three panels on the right side, are observed to be more polarized in comparison to the optimal distributions. For example, in all cases, creators generate content from the drama genre (the most popular genre) more than the ideal frequency. Such polarization effect is alleviated as the exploration factor on the creator side increases. 
Interestingly, if we look at the empirical frequency under different $n$ for a fixed $\epsilon$, we can find that the content distribution also becomes increasingly diversified as $n$ gets larger, which is beneficial for social welfare. This explains our finding in Table \ref{tb:exp3_large_n_beta0.1} and \ref{tb:exp3_large_n_beta0.5} as well: that PotA gets consistently smaller under various $\epsilon$ as the number of creators increases. These observations provide insight into the positive impact of creators adopting regret minimization strategies: the resulting social welfare is fairly good, even if they are subject to bounded rationalities.

\begin{figure}[t]
\centering
\includegraphics[width=1\columnwidth]{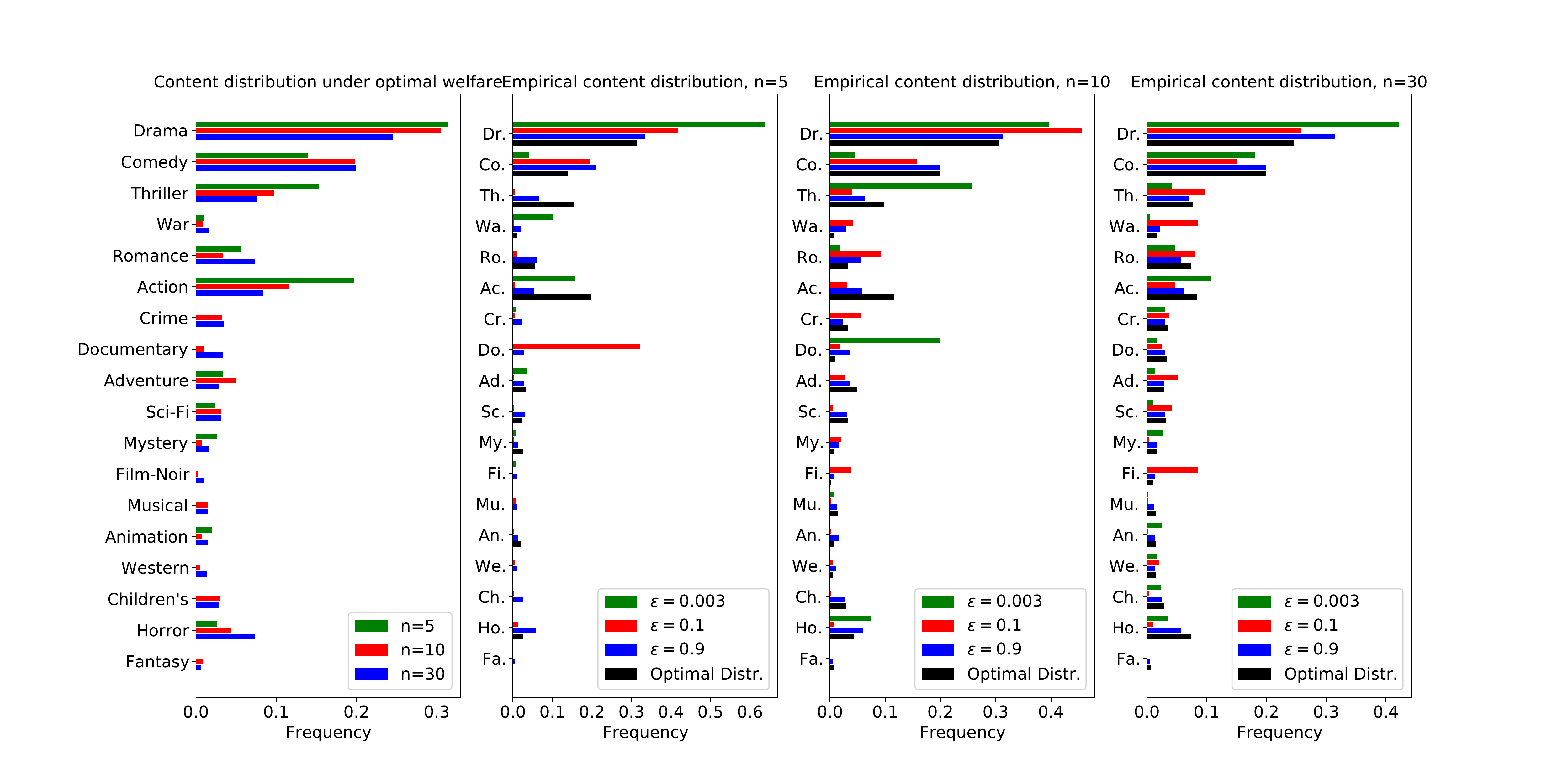}
\caption{The genre distribution induced by creators' adaptive content creation strategies under different $n$. The left most panel shows the optimal genre distribution of the joint content creation strategy that maximizes social welfare under $n=5,10,30$. The remaining 3 panels from left to right correspond to the genre distribution obtained from simulation under $n=5,10,30$. The green/red/blue colors in the 3 panels on the right represent simulation results with different exploration rate $\epsilon =0.003, 0.1, 0.9$ and the black bars denote the optimal distribution. All the frequency values are averaged over 5 independent runs. }
\label{fig:movieLens_genre_distribution}
\end{figure}
\section{Conclusion and Future Work}

We propose the \game{}, a game-theoretical framework for analyzing the strategic behaviors of creators on online content recommendation platforms. Our primary contribution is a comprehensive characterization of social welfare as the outcome of competition among creators, which suggests that the traditional top-$K$ recommendation principle is effective when the platform utilizes user engagement as an incentive metric and offers a sufficient number of choices to users, resonating with the well-known ``invisible hand" argument posited by Adam Smith.


Our positive result hinges on two crucial assumptions: 1. Users on the platform employ the RU model with Gumbel noise when making decisions; 2. the incentive provided by the platform in creator utilities must align with the user utilities (i.e., user engagement) utlized in the social welfare metric. A fascinating and significant future direction may require the relaxation of either of these assumptions. Consequently, we propose two open questions:

{\bf Open Question 1: } Does our main result still hold under different user choice models? For instance, what if users are permitted to make multiple choices, or users adhere to the RU model but with a different decision noise (e.g., Gaussian)? Although we have some preliminary empirical evidence suggesting a similar PoA upper bounded strictly less than $2$, rigorous theoretical analysis can be challenging, as players' utility functions do not necessarily have closed forms and we might need novel techniques to derive a new smoothness constant.

{\bf Open Question 2: } What can we say regarding the social welfare when user utility metrics employed in players' incentives and social welfare do not align? Although it is reasonable for the platform to design such aligned incentives, practical reasons might occasionally render this challenging to implement. Although our Proposition \ref{prop:unbound_user_welfare} demonstrates that the welfare can be arbitrarily bad if player incentives do not correspond with user engagement, it is still intriguing to examine the intermediate situations: what if player incentives and user engagement are correlated but not perfectly aligned? Any supplementary understanding of this extended setting can yield more profound impact on real-world systems.

\section{Acknowledgement}
This work is supported in part by the US National Science Foundation under grants IIS-2007492, IIS-2128019 and IIS-1838615. Haifeng Xu is supported in part by an ARO award W911NF-23-1-0030. 

\vskip 0.2in

\bibliography{main}

\appendix 
\newpage


\section{Proof of Theorem \ref{thm:PoA}}\label{app:poa_proof}
First we derive the closed-forms of the utility and welfare functions of \game{} given in Lemma \ref{lm:closed_form_utility}.

\subsection{Proof of Lemma \ref{lm:closed_form_utility}:  Closed Forms of   Utility and Welfare Functions}\label{app:close_utility}
We start with a few known and useful properties of Gumbel distributions. 
    \begin{lemma}\label{lm:gumbel_property}[e.g., \citep{balog2017lost}]
Let $(v_1,\cdots,v_n)\in\RR^n$ be any real-valued vector and $\varepsilon_1, \cdots,\varepsilon_n$ be independent samples from $\text{Gumbel}(\mu,\beta)$. Then
    \begin{equation}
        \arg\max_i(v_i+\varepsilon_i) \sim \text{Categorical}\Big(\frac{\exp(\beta^{-1}v_i)}{\sum_{j=1}^n \exp(\beta^{-1}v_j)}\Big),
    \end{equation}
    and     
    \begin{equation}
        \max_i(v_i+\varepsilon_i) \sim \text{Gumbel}\Big(\mu+\beta\log \Big(\sum_{j=1}^n \exp(\beta^{-1}v_j)\Big), \beta\Big).
    \end{equation}
\end{lemma}
    
\vspace{2mm}
\noindent
\textbf{Derivation of user utility and welfare.}
These derivations follow easily from Lemma \ref{lm:gumbel_property}. Since we assumed that $\varepsilon_i \sim $Gumbel$(-\beta\gamma,\beta)$, leveraging properties in Lemma \ref{lm:gumbel_property} we conclude that $\x_j$'s choice distribution over $K$ alternatives $\{\s_1,\cdots,\s_K\}=\T_j(\s;K)$ follows the soft-max rule
    \begin{equation}
        \Pr[\x_j \shortrightarrow \s_i] = \frac{\exp(\beta^{-1}\sigma(\s_i, \x_j))}{\sum_{\s_k\in\T_j(\s;K)} \exp(\beta^{-1}\sigma(\s_k, \x_j))},
    \end{equation}

and the expected user utility after making choices has the following form
    \begin{align}\label{eq:746}
      \pi_j(\x_j) =   \EE\left[\max_{i\in [K]} \{\sigma(\s_i, \x_j)+\varepsilon_i\}\right] &= \beta\log\left[\sum_{\s_k\in\T_j(\s;K)} \exp(\beta^{-1}\sigma(\s_k, \x_j))\right].
    \end{align}
Taking expectation over all users, we obtain the following welfare function
\begin{align}\label{eq:G_welfare}
    W(\s) = \sum_{j=1}^m \EE\left[\max_{\s_k\in\T_j(\s;K)} \{\sigma(\s_k, \x_j)+\varepsilon_i\}\right] 
    =\beta\sum_{j=1}^m \log \left[\sum_{\s_k\in\T_j(\s;K)} \exp{(\beta^{-1}\sigma(\s_k, \x_j))}\right].
\end{align}

By setting $\tilde{W}(\s)=\beta W(\s), \tilde{\sigma}(\s,\x)=\beta^{-1}\sigma(\s,\x)$, we have $$\tilde{W}(\s)=\sum_{j=1}^m \log [\sum_{\s_k\in\T_j(\s;K)} \exp{(\tilde{\sigma}(\s_k, \x_j))}].$$ Therefore, under a rescaling of constant $\beta$ it is with out loss of generality to consider a scoring function $\sigma \in [0, \frac{1}{\beta}]$, the user utility function and the social welfare function in the following form

\begin{equation}\label{eq:user_utility_close}
        \pi_j(\s)= \log \Big[\sum_{\s_k\in\T_j(\s;K)} \exp{(\sigma(\s_k, \x_j))}\Big],
\end{equation}

\begin{equation}\label{eq:welfare_close}
        W(\s)= \sum_{j=1}^m \log \Big[\sum_{\s_k\in\T_j(\s;K)} \exp{(\sigma(\s_k, \x_j))}\Big].
\end{equation}

\vspace{2mm}
\noindent
\textbf{Derivation of creator utility.}
This turns out to be a new result which requires non-trivial arguments. The players' utility is   given by
\begin{align}\label{eq:756}
        u_i(\s)&=\sum_{j=1}^m \EE[\sigma(\s_i, \x_j)+\varepsilon_i|\x_j \shortrightarrow \s_i]\cdot \Pr[\x_j \shortrightarrow \s_i] \\ \label{eq:G_util}
        &= \sum_{j=1}^m \EE[\sigma(\s_i, \x_j)+\varepsilon_i|\x_j \shortrightarrow \s_i]\cdot \frac{\exp(\sigma(\s_i, \x_j))}{\sum_{\s_k\in\T_j(\s;K)} \exp(\sigma(\s_k, \x_j))},
\end{align} 
 
   According to the definition in \eqref{eq:G_util}, what we need to show is that for i.i.d. random variables $\{\varepsilon_i\}_{i=1}^K$ sampled from Gumbel$(-\beta\gamma,\beta)$, 
    \begin{equation}
        \EE[\sigma(\s_i, \x_j)+\varepsilon_i|\x_j \shortrightarrow \s_i] = \EE[\max_{k\in[K]}\{\sigma(\s_k, \x_j)+\varepsilon_i\}]=\log \Big[\sum_{\s_k\in\T_j(\s;K)} \exp{(\sigma(\s_k, \x_j))}\Big],
\end{equation}
i.e., for any $(v_1,\cdots,v_K)\in \RR^K$ and i.i.d. random variables $\{\varepsilon_i\}_{i=1}^K$ sampled from Gumbel$(0,1)$,

\begin{equation}\label{eq:790}
    \EE[v_i+\varepsilon_i|i=\arg\max_{k\in[K]}(v_k+\varepsilon_k)]=\gamma+\log\Big(\sum_{k=1}^K \exp(v_k)\Big).
\end{equation}

Let $Y_i=\max_{k\in [K],k\neq i}(v_k+\varepsilon_k)\sim \text{Gumbel}(\log(\sum_{k\neq i}\exp(v_k)),1)$ and $X_i=v_i+\varepsilon_i \sim \text{Gumbel}(v_i, 1)$. Then $X_i$ has the probability density function
\begin{equation}
    f_i(x)=\exp(-((x-v_i)+e^{-(x-v_i)})),
\end{equation}
and $Y$ has the cumulative distribution function
\begin{equation}
    F_i(y)=\exp(-e^{-(y-\log(\sum_{k\neq i}\exp(v_k)))})).
\end{equation}
Therefore we can explicitly compute the conditional expectation of $X_i$ as follows: 
\begin{align}\notag
   &\EE[v_i+\varepsilon_i|i=\arg\max_{k\in[K]}(v_k+\varepsilon_k)] \\\notag
   =& \EE[v_i+\varepsilon_i|v_i+\varepsilon_i \geq \max_{k\in [K],k\neq i}(v_k+\varepsilon_k)] \\ \label{eq:812}
   =& \EE[X|X \geq Y, X\sim \text{Gumbel}(v_i, 1), Y\sim \text{Gumbel}(\log(\sum_{k\neq i}\exp(v_k)),1)] \\\notag
   =& \frac{\int_{\RR} xf_i(x) F_i(x)dx}{\int_{\RR} f_i(x) F_i(x)dx} \\\notag
   =& \frac{\int_{\RR} x\exp(-((x-v_i)+e^{-(x-v_i)})) \exp(-e^{-(x-\log(\sum_{k\neq i}\exp(v_k)))}))dx}{\int_{\RR} \exp(-((x-v_i)+e^{-(x-v_i)})) \exp(-e^{-(x-\log(\sum_{k\neq i}\exp(v_k)))}))dx} \\ \label{eq:809}
   =& \frac{\int_{\RR_{\geq 0}} -\ln t\cdot \exp (-t\sum_{k=1}^K\exp(v_k))dt}{\int_{\RR_{\geq 0}} \exp (-t\sum_{k=1}^K \exp(v_k)) dt} \\ \label{eq:810}
   =& \ln\Big(\sum_{k=1}^K \exp(v_k))\Big) + \frac{\int_{\RR_{\geq 0}} -\ln s\cdot \exp (-s)ds}{\int_{\RR_{\geq 0}} \exp (-s) ds} \\ \notag 
   =& \ln\Big(\sum_{k=1}^K\exp(v_k))\Big) - \frac{d}{d\alpha}\int_{\RR_{\geq 0}} s^{\alpha}e^{-s}ds \\ \notag 
   =& \ln\Big(\sum_{k=1}^K \exp(v_k))\Big) - \frac{d}{d\alpha}\Gamma(\alpha+1)\Big|_{\alpha=0} \\ \label{eq:813}
   =& \ln\Big(\sum_{k=1}^K \exp(v_k))\Big) + \gamma.
\end{align}
where \eqref{eq:812} holds because of Lemma \ref{lm:gumbel_property}, \eqref{eq:809} and \eqref{eq:810} hold by change of variables $t=e^{-x}$ and $s=t\sum_{k=1}^K \exp(v_k))$, and \eqref{eq:813} is from the definition of Euler-Mascheroni constant. Therefore we show \eqref{eq:790} and the players' utility function has the following form

\begin{equation}\label{eq:player_utility_close}
        u_i(\s)= \sum_{j=1}^m \Big(\log \Big[\sum_{\s_k\in\T_j(\s;K)} \exp{(\sigma(\s_k, \x_j))}\Big]\Big) \frac{\II[\s_i\in \T_j(\s;K)]\exp(\sigma(\s_i, \x_j))}{\sum_{\s_k\in \T_j(\s;K)} \exp(\sigma(\s_k, \x_j))}.
\end{equation}

\subsection{Proof of Lemma \ref{lm:submodular_V_} and \ref{lm:utility_property_}: Properties of Utility and Welfare Functions }\label{app:lemma23_proof}

We consider the utility and welfare functions given in \eqref{eq:user_utility_close}, \eqref{eq:player_utility_close} and \eqref{eq:welfare_close} under the re-scaling of constant $\beta$ with the new assumption that $\sigma(\s,\x) \in [0, \frac{1}{\beta}], \forall \s\in \cup_{i=1}^n \S_i, \x\in \X$. To simplify the subsequent analysis, we first specify some useful notations and conventions. For any joint strategy profile $\s=(\s_1, \cdots, \s_n)$, we use capital letter $S$ to denote its set representation, i.e., $S=\{\s_1, \cdots, \s_n\}$. In this way we can view $\T_j(\s;K),\pi_j(\s), u_i(\s),W(\s)$ defined in \eqref{eq:Tj}, \eqref{eq:user_utility_original}, \eqref{eq:player_utility_original}, \eqref{eq:welfare_original} as set functions $\T_j(S;K), \pi_j(S), u_i(S), W(S)$. From now on, we will use the set notation $S$ and the vector notation $\s$ interchangeably, depending on the context. Similarly, we use $S_{-i}$ to denote the set $\{\s_1, \cdots, \s_n\}$ excluding element $\s_i$. Moreover, we extend the definition of $\T_j(S;K)$ by allowing $|S|=K-1$ in the following sense: when $|S|=K-1$, we let $\T_j(S;K)=S \cup \{\bar{\s}\}$, where $\bar{\s}$ is a default external choice such that $\sigma(\bar{\s}, \x)=0$ for all $\x\in \X$. This extension captures the situation when the system does not have enough active content creators to allocate to the users. When such a situation happens, the system will put a default choice $\bar{\s}$ in the top-$K$ list without any utility guarantee. We remark that this extended definition is introduced merely for the convenience of presentation and does not affect the implication of our main result. 

Prior to the proofs for Lemma \ref{lm:submodular_V_} and \ref{lm:utility_property_}, we present two intermediate results in Proposition \ref{prop:monotone_V} and Lemma \ref{lm:eq_monotone}. Proposition \ref{prop:monotone_V} reveals a rather basic property of social welfare $W$ which is useful in the proof of Theorem \ref{thm:PoA}, and Lemma \ref{lm:eq_monotone} is useful in the proof of Lemma \ref{lm:utility_property_}.

\begin{proposition}\label{prop:monotone_V}
Fix a joint strategy $S=\{\s_1, ..., \s_n\}$ in any $n-$player \game{} $\G$. If we add an additional player indexed by $n+1$ with pure strategy $\s_{n+1}$ to the game and let $S'=\{\s_1, ..., \s_n, \s_{n+1}\}$, the social welfare $W$ will strictly increase, i.e.,
    \begin{equation}\label{eq:772}
    W(S') > W(S).
    \end{equation}
\end{proposition}

\begin{proof}
By definition,
\begin{equation}
    W(S') = \sum_{j=1}^m \Big(\log \Big[\sum_{\s\in \T_j(S';K)} \exp{(\sigma(\s, \x_j))}\Big]\Big),
\end{equation}
\begin{equation}
    W(S) = \sum_{j=1}^m \Big(\log \Big[\sum_{\s\in \T_j(S;K)} \exp{(\sigma(\s, \x_j))}\Big]\Big),
\end{equation}
 It is obvious that for any fixed user $j$, the sum of exponential scores of top-$K$ choices from $S'$ is better than that from $\S$, i.e.,
$$\sum_{\s\in \T_j(S';K)} \exp{(\sigma(\s, \x_j))}\geq \sum_{\s\in \T_j(S';K)} \exp{(\sigma(\s, \x_j))}.$$
Therefore, \eqref{eq:772} holds immediately by the monotonicity of the logarithmic function.
\end{proof}

 Proposition \ref{prop:monotone_V} reveals an important yet natural property of real-world content provider competitions: when there are more competitors in the market, users are facing more alternatives and thus their welfare will always increase. 




\begin{lemma}\label{lm:eq_monotone}
The following function
\begin{equation}\label{eq:monotone_xy}
    f(x,y)=\frac{(x+1)\log(x+y)}{(x+y)(\log(x+y)-\log y)}, (x,y)\in \RR_{+}\times \NN_+,
\end{equation}
is monotonically increasing in $y$ for any $x \in \RR_{+}$, and is monotonically decreasing in $x$ for any integer $y \in \NN_+$.
\end{lemma}
\begin{proof}
We first demonstrate the monotonicity of $f(\cdot,y)$ by directly calculating its partial derivatives. Note that $t\geq \log(1+t)$ holds for any $t\geq0$, we have 
\begin{equation}
    \frac{1}{x+1}\frac{\partial f(x,y)}{\partial y}=\frac{\log(1+\frac{x}{y})+\log(x+y)[\frac{x}{y}-\log(1+\frac{x}{y})]}{[(x+y)(\log(x+y)-\log y)]^2}>0,
\end{equation}
which implies that $f(x,y)$ is increasing in $y$. Now it remains to show the monotonicity w.r.t. $x$ when fixing $y=K$, which is slightly more intricate. The derivative of $f(x,K)$ w.r.t. $x$ now writes
\begin{align*} 
   f'(x,K) &= \frac{(K-1)\log(x+K)\log(1+\frac{x}{K})-(x+1)\log K}{[(x+K)(\log(x+K)-\log K)]^2} \\ &\triangleq \frac{-g(x,K)}{[(x+K)(\log(x+K)-\log K)]^2},
\end{align*}
and 
\begin{align}\label{eq:gxy103}
    g'(x,K) &= \frac{1}{x+K}\Big[(2K+x-1)\log K-2(K-1)\log(x+K)\Big] \\ \notag
    & = \frac{2(K-1)}{x+K} \Big[\frac{2K+x-1}{2(K-1)}\log K-\log(x+K)\Big] \\ \label{eq:gxy105}
    & = \frac{2(K-1)}{x+K} \Big[\frac{x+1}{2(K-1)}\log K-\log(1+\frac{x}{K})\Big] \\ \notag
    & \geq \frac{2(K-1)}{x+K} \Big[\frac{x+1}{2(K-1)}\log K-\frac{x}{K}\Big]  \\\label{eq:gxy107}
    & \geq \frac{x}{x+K} \Big[\log K-\frac{2(K-1)}{K}\Big].
\end{align}

We claim $g'(x,K)
\geq 0, \forall K\in \mathbb{N}^+$, and this is because
\begin{enumerate}
    \item if $K=1$, from \eqref{eq:gxy103} we have $g'(x,K)=0$.
    \item for $K\geq 5$, we can verify $\log K-\frac{2(K-1)}{K}>0$. From \eqref{eq:gxy107} we have $g'(x,K)>0$.
    \item for $K\in \{2,3,4\}$, we can verify $\frac{x+1}{2(K-1)}\log K-\log(1+\frac{x}{K})>0$ for any $x\geq 0$. Therefore \eqref{eq:gxy105} we have $g'(x,K)>0$.
\end{enumerate}
Now since $g'(x,K)\geq 0$, we conclude that $g(x,K)\geq g(0,K)=\log K \geq 0$, which implies $f'(x,K)\leq 0, \forall K\in \mathbb{N}^+$. Hence, $f(x,K)$ is decreasing in $x$.
\end{proof}

Now we are ready to prove Lemma \ref{lm:submodular_V_} and \ref{lm:utility_property_}.

\begin{proof}[\textbf{of Lemma \ref{lm:submodular_V_}}]
By the definition we only need to show the submodularity of $\pi_j(S)$ for any $j\in [m]$, i.e., 
\begin{equation}\label{eq:submodular29}
        \pi_j(\T_j(S\cup \{\s_x\};K))-\pi_j(\T_j(S;K)) \geq \pi_j(\T_j(S\cup \{\s_x,\s_y\};K))-\pi_j(\T_j(S\cup \{\s_y\};K)).
\end{equation}

With out loss of generality we assume $\sigma(\s_x,\x_j)\geq \sigma(\s_y,\x_j)$, and let $$\{v_1,\cdots, v_K\}=\{\exp(\sigma(\s,\x_j))|\s \in \T_j(S;K)\},$$
where $v_1\leq \cdots \leq v_K$. Then depending on the values of $v_x=\exp(\sigma(\s_x,\x_j)),v_y=\exp(\sigma(\s_y,\x_j))$ and $K$, there are three situations :
\begin{enumerate}
    \item $v_x\leq v_1$: \eqref{eq:submodular29} holds because its LHS and RHS are both equal to $0$.
    \item $v_x > v_1, K=1$: The LHS of \eqref{eq:submodular29} is equal to $\log \frac{v_x}{v_1}>0$, the RHS of \eqref{eq:submodular29} is equal to 0. 
    \item $v_x > v_1, K\geq 2$: The LHS of \eqref{eq:submodular29} is equal to $\log \frac{v_x+v_2+a}{v_1+v_2+a}$, the RHS of \eqref{eq:submodular29} is equal to $\log \frac{v_x+v_y+a}{v_y+v_2+a}$, where $a=\sum_{k=3}^K v_k$ if $K\geq 3$ and $a=0$ if $K=2$. We can verify 
    \begin{align*}
        & (v_x+v_2+a)(v_y+v_2+a)-(v_1+v_2+a)(v_x+v_y+a) \\= &(v_2-v_1)(a+v_1+v_2)+(v_x-v_1)(v_y-v_1)\geq 0.
    \end{align*}
\end{enumerate}
Therefore, \eqref{eq:submodular29} holds and Lemma \ref{lm:submodular_V_} follows by summing \eqref{eq:submodular29} over all $j\in [m]$. 
\end{proof}

\begin{proof}[\textbf{of Lemma \ref{lm:utility_property_}}]
By definition, 
\begin{equation}
        u_i(\s_i;\s_{-i})= \sum_{j=1}^m \Big(\log \Big[\sum_{\s'\in \T_j(S;K)} \exp{(\sigma(\s', \x_j))}\Big]\Big) \frac{\II[\s_i\in \T_j(S;K)]\exp(\sigma(\s_i, \x_j))}{\sum_{\s'\in \T_j(S;K)} \exp(\sigma(\s', \x_j))},
\end{equation}
and 
\begin{equation}
    W(S) = \sum_{j=1}^m \pi_j(\T_j(S;K)).
\end{equation}

It is sufficient to prove that for any user $j$,

\begin{align}\notag
    & \Big(\log \Big[\sum_{\s'\in \T_j(S;K)} \exp{(\sigma(\s', \x_j))}\Big]\Big) \frac{\II[\s_i\in \T_j(S;K)]\exp(\sigma(\s_i, \x_j))}{\sum_{\s'\in \T_j(S;K)} \exp(\sigma(\s', \x_j))} \\\label{eq:utility_property_67} &\geq c(\beta,K)\cdot\big[  \pi_j(\T_j(S;K))-\pi_j(\T_j(S_{-i};K))\big].
\end{align}

Note that when $\s_i\notin \T_j(S;K)$, \eqref{eq:utility_property_67} is trivial as its LHS=RHS=0. Now we suppose $\s_i\in \T_j(S;K)$ and thus $\T_j(S_{-i};K)$ and $\T_j(;K)$ only differ in one element. Without loss of generality we let 
$$\{\exp(\sigma(\s,\x_j))|\s \in \T_j(S_{-i};K)\}=\{v'_1,v_2,\cdots, v_K\},$$
and
$$\{\exp(\sigma(\s,\x_j))|\s \in \T_j(S;K)\}=\{v_1,v_2,\cdots, v_K\}, v_1\geq v'_1.$$

Because of our extended definition of $\T_j(S;K)$, $\T_j(S_{-i},K)$ is well defined when $K=n$, under which case we have $v'_1=\exp(\sigma(\bar{\s},\x_j))=1$. Now we let $z=v_2+\cdots+v_K$, \eqref{eq:utility_property_67} is equivalent to 
\begin{equation}\label{eq:77}
    \frac{v_1}{v_1+z}\log(v_1+z)\geq c(\beta,K) \cdot \log\Big[\frac{v_1+z}{v'_1+z}\Big].
\end{equation}
Since $v'_1=\exp(\sigma(\cdot,\x_j))\geq 1$, a sufficient condition for \eqref{eq:77} to hold is
\begin{equation}\label{eq:81}
    \frac{v_1}{v_1+z}\cdot \frac{\log(v_1+z)}{\log(v_1+z)-\log(1+z)}\geq c(\beta,K) .
\end{equation}

Note that $x=v_1-1 \in [0, e^{1/\beta}-1], y = z+1 \in [K, (K-1)e^{1/\beta}+1]$, the LHS of \eqref{eq:81} becomes a function of $(x,y)$ which has the following form
\begin{equation}
    f(x,y)=\frac{(x+1)\log(x+y)}{(x+y)(\log(x+y)-\log y)}.
\end{equation}
From Lemma \ref{lm:eq_monotone} we know $f(x,y)$ is monotonically increasing in $y$ for any $x>0$ and is monotonically decreasing in $x$ any integer $K\geq 1$. Therefore, it holds that 
    \begin{equation}\label{eq:cbetaK_minima}
        f(x,y)\geq f(x,K)\geq f(e^{1/\beta}-1, K)=c(\beta, K).
    \end{equation}

Hence, \eqref{eq:utility_property_67} holds and we complete the proof.
\end{proof}

\subsection{Proof of Theorem \ref{thm:PoA}}

With the help of Proposition \ref{prop:monotone_V}, Lemma \ref{lm:utility_property_} and \ref{lm:submodular_V_}, now we are ready to prove our claim in Theorem \ref{thm:PoA}. We will demonstrate that any \game{} instance $\G(\{\S_i\}_{i=1}^n, \X, \sigma, \beta, K)$ is a smooth game with parameter $(\lambda, \mu)=(c(\beta, K), c(\beta, K))$ so that its PoA can be upper bounded by $\frac{1 + \mu}{\lambda}=1+\frac{1}{c(\beta, K)}$.

\begin{proof}
Let $\s=(\s_1,...,\s_n)$ and $\s^*=(\s_1^*,...,\s_n^*)$ be two different strategy profiles. First, due to function $W$'s sub-modular property disclosed in Lemma \ref{lm:submodular_V_}, for every $i\in [n]$ we have 
\begin{equation}
    W([\s_i^*,\s_{-i}]) - W(\s_{-i}) \geq W([\s_1^*, \cdots, \s_{i-1}^*,\s_i^*,\s]) - W([\s_1^*, \cdots, \s_{i-1}^*,\s]).
\end{equation}
Summing over all player $i$ we obtain
\begin{align}\notag
   \sum_{i=1}^n (W([\s_i^*,\s_{-i}]) - W(\s_{-i})) &\geq \sum_{i=1}^n (W([\s_1^*, \cdots, \s_{i-1}^*,\s_i^*,\s]) - W([\s_1^*, \cdots, \s_{i-1}^*,\s])) \\ \notag
   & = W([\s^*,\s])-W(\s) \\  \label{eq:165}
   & > W(\s^*)-W(\s),
\end{align}
where the last inequality holds because of Proposition \ref{prop:monotone_V}. On the other hand, from Lemma \ref{lm:utility_property_} it also holds that 
\begin{equation}\label{eq:169}
    u_i(\s_i^*;\s_{-i}) \geq c(\beta, K)\cdot\big[W([\s_i^*,\s_{-i}]) - W(\s_{-i})\big],
\end{equation}

And therefore
\begin{align} 
\label{eq:174}
   \sum_{i=1}^n u_i(\s_i^*;\s_{-i}) & \geq c(\beta, K)\cdot \sum_{i=1}^n\big[W([\s_i^*,\s_{-i}]) - W(\s_{-i})\big] \\ \label{eq:175}
    & > c(\beta, K) [W(\s^*)-W(\s)].
\end{align}
where inequality \eqref{eq:174} holds by \eqref{eq:169}, and inequality \eqref{eq:175} holds by \eqref{eq:165}. 

Since \eqref{eq:175} holds for any $\s\in\S$, for any $\al\in CCE(\G)$ we can take expectation over $\s \sim \al$ and obtain

\begin{equation} 
   \sum_{i=1}^n \EE_{\s\sim \al}[u_i(\s_i^*;\s_{-i})] > c(\beta, K) [W(\s^*)-\EE_{\s\sim \al}[W(\s)]].
\end{equation}

Therefore,
\begin{align} \notag
    \EE_{\s\sim \al}[W(\s)] &= \EE_{\s\sim \al}[\sum_{i=1}^n u_i(\s)] \\ \label{eq:173}
    & \geq \EE_{\s\sim \al}[\sum_{i=1}^n u_i(\s_i^*;\s_{-i})] \\  \notag
    & \geq c(\beta, K)\cdot \sum_{i=1}^n\big[\EE_{\s\sim \al}[W([\s_i^*,\s_{-i}])] - \EE_{\s\sim \al}[W(\s_{-i})]\big] \\ \label{eq:1755}
    & > c(\beta, K) [W(\s^*)-\EE_{\s\sim \al}[W(\s)]].
\end{align}
where inequality \eqref{eq:173} follows by the definition of CCE and inequality \eqref{eq:1755} holds by \eqref{eq:175}. Rearranging terms we obtain 

\begin{equation}
    PoA(\G)=\frac{\max_{\s\in\S}W(\s)}{\min_{\al\in CCE(\G)}\EE_{\s\sim \al}[W(\s)]}< 1+\frac{1}{c(\beta,K)}.
\end{equation}
\end{proof}

\subsection{Proof of the Property of $c(\beta, K)$}
\begin{proof}
The $c(\beta, K)$ function has the following form:
\begin{align}
c(\beta, K)= \frac{(b+1)\log(b+K)}{(b+K)(\log(b+K)-\log K)}, b=e^{\frac{1}{\beta}}-1.
\end{align}  
We prove the following facts one by one.

\begin{enumerate}
    \item Fix any $\beta>0$, $c(\beta, K)$ is monotonically increasing in $K$; similarly, fix any $K\geq 1$, $c(\beta, K)$ is monotonically increasing in $\beta$.
    
    Note that $e^{\frac{1}{\beta}}-1$ is decreasing in $\beta$, from Lemma \ref{lm:eq_monotone} the claim holds.

    \item $c(\beta, K)=1$ if and only if $K=1$ or $\beta \rightarrow 0$.

    When $K=1$, $c(\beta, K)=1$ directly holds. When $\beta \rightarrow 0$, $b\rightarrow +\infty$ and $c(\beta, K) \rightarrow 1$. The ``only if" direction follows from the monotonicity property of $c(\beta, K)$.
    
    \item For any $\beta>0$ and $K\geq 1$, we have $c(\beta, K)\geq 1$ and thus $PoA(\G) < 2$ always holds.

    By the monotonicity of $c$,  $c(\beta, K)\geq c(\beta, 1)=1$. Hence, $PoA(\G)<1+\frac{1}{c(\beta, K)}\leq 2$.

    \item For sufficiently large $\beta$ and $K$, $c(\beta, K) \approx (1+\beta)\log K$ asymptotically, and therefore 
    \begin{equation}
        PoA(\G) < 1+\frac{1}{(1+\beta)\log K}.
    \end{equation} 

    When $\beta$ is sufficiently large, $b=e^{\frac{1}{\beta}}-1 \approx \frac{1}{\beta}\rightarrow 0$. Therefore, 

    \begin{align*}
        c(\beta, K) &= \frac{(b+1)\log(b+K)}{(b+K)\log(1+\frac{b}{K})} \\ 
        & \approx \frac{(b+1)K\log(b+K)}{(b+K)b} & \mbox{since $\log (1+x)\approx x \text{ as } x\rightarrow 0$ }\\ 
        & \approx \frac{(b+1)\log K}{b} & \mbox{since $K>>b$} \\
        & \approx (1+\beta)\log K. & \mbox{since $b\approx \frac{1}{\beta}$}
    \end{align*}
\end{enumerate}

\end{proof}

\section{Proof of Corollary \ref{coro:noregret_CCE}}\label{app:poa_proof_regret}
\begin{proof}
Let $\epsilon(T)=\frac{R(T)}{T}$, and $\s^*=(\s_1^*,...,\s_n^*)$ be a global maximizer of $W(\s)$. By definition, 
\begin{equation}
    \EE_{\s\sim\al}[u_i(\s)] \geq \EE_{\s\sim\al}[u_i(\s^*_i,\s_{-i})] - \epsilon(T).
\end{equation}
Summing over all player $i\in [n]$ we obtain 
\begin{equation}\label{eq:129}
    \EE_{\s\sim\al}[W(\s)]=\sum_{i=1}^n\EE_{\s\sim\al}[u_i(\s)] \geq \sum_{i=1}^n\EE_{\s\sim\al}[u_i(\s^*_i,\s_{-i})] - n\epsilon(T).
\end{equation}

On the other hand, by \eqref{eq:175} from the proof of Theorem \ref{thm:PoA}, we have 
\begin{equation}
    \sum_{i=1}^n u_i(\s_i^*;\s_{-i}) > c(\beta, K) [W(\s^*)-W(\s)], \forall \s \in S.
\end{equation}
Taking the expectation of $\s$ over distribution $\al$ we obtain
\begin{equation}\label{eq：138}
    \sum_{i=1}^n \EE_{\s\sim\al}[u_i(\s_i^*;\s_{-i})] > c(\beta, K) \big(W(\s^*)- \EE_{\s\sim\al}[W(\s)]\big).
\end{equation}

\eqref{eq:129} and \eqref{eq：138} together imply that 
\begin{equation}\label{eq:142}
    \EE_{\s\sim\al}[W(\s)] +n\epsilon(T)> c(\beta, K) \big(W(\s^*)- \EE_{\s\sim\al}[W(\s)]\big).
\end{equation}

Note that for any $\s\in\S$, we have $W(\s)=\sum_{j=1}^m \Big(\log \Big[\sum_{\s\in \T_j(\s;K)} \exp{(\sigma(\s, \x_j))}\Big]\Big)\geq \beta\log K$ and therefore, $n\epsilon(T) \leq \frac{n\epsilon(T)}{\beta\log K} \EE_{\s\sim\al}[W(\s)]$. Substituting it into \eqref{eq:142}, we obtain \eqref{eq:PoA_noregret}.

\end{proof}

\section{Proof of Corollary \ref{prop:welfare_loss}}\label{app:welfare_loss}

\begin{proof}
Note that fix any players' strategy profile $\s$, the top-$K$ matching mechanism maximizes the social welfare $W$. Therefore, $W(\G')\leq W^*_K$. On the other hand, from the PoA bound in Theorem \ref{thm:PoA} it holds that 
\begin{align*}
    \frac{W^*_K}{W(\G)} &< 1+ \frac{(b+K)(\log (1+b/K))}{(b+1)\log(b+K)} \\ 
    & < 1+ \frac{(b+K)(b/K)}{(b+1)\log(b+K)} \\ 
    & = 1 + \frac{b+K}{K\log (b+K)}.
\end{align*}

Rearranging term yields $W(\G')-W(\G) \leq W^*_K-W(\G)\leq W^*_K / (1+\frac{K\log (b+K)}{b+K})$.
\end{proof}

\section{Proof of Theorem \ref{thm:PoA_lowerbound}}\label{app:poa_lower_proof}
\begin{proof}
Let $b=\exp(1/\beta)-1$. Consider an $n$-player game where each player-$i$ has the same action set $\S_i=\{\x_1, \cdots, \x_n\}$. Let the user population $\X$ be a set with size $m=n+(n-1)a$, in which $n$ users have profile $\x_1$ and $a$ users have profile $\x_i$ for $i=2,\cdots, n$. Here $a=\beta\log K +1$ is a constant whose choice will become clear later. Let the scoring function $\sigma$ be the indicator function defined as follow: 

\begin{equation}  
\sigma(\s, \x)=\left\{  
             \begin{array}{lr}  
             1, \text{~~if~~} \s=\x,&   \\  
             0, \text{~~otherwise.}&    
             \end{array}  
\right.  
\end{equation} 

First we lower bound the optimal welfare $\max_{\s\in\S}W(\s)$. Consider the joint-strategy profile $\s^*=(\x_1, \x_2, \cdots, \x_n)$, under which each user gets one player with $\sigma$ score $1$ and $K-1$ player with $\sigma$ score 0. In this case, each user-$j$'s utility $\pi_j(\s)=\beta\log (b+K)$ and the social welfare $W(\s)=m\beta\log (b+K)$. Therefore, the optimal social welfare 
\begin{equation}  \label{eq:125}
\max_{\s\in \S}W(\s)\geq W(\s^*)=m\beta\log (b+K).    
\end{equation}

Next we show that $\s=(\x_1,\x_1,\cdots,\x_1)$ is a pure NE of $\G$ and thus $\s\in CCE(\G)$. Given players' joint-strategy $\s$, $n$ users will be assigned with $K$ players with $\sigma$ score $1$ and $(n-1)a$ users will be assigned with $K$ players with $\sigma$ score $0$. Therefore, the utility for an arbitrary player-$i$ is given by
\begin{align}\notag
        u_i(\s) &= \Big[n\cdot(\beta\log K + 1) + a(n-1)\cdot \beta \log K\Big]/n \\  \label{eq:131}
        & = \beta\log K + 1 + \frac{a(n-1)\beta \log K}{n}.
\end{align}

If player-$i$ switches from strategy $\s_1$ to $\s_j$, $n$ users still get $K$ players with score $1$, $(n-2)a$ users get players with score $0$, and $a$ users get $K$ players with scores $(1, 0,\cdots, 0)$. Therefore, player-$i$'s utility after the deviation is 

\begin{align}\notag
        u_i(\s_j, \s_{-i}) &= n\cdot0 + a(n-2)\cdot \beta \log K\cdot \frac{1}{n} + a\cdot \beta\log (b+K)\cdot \frac{e^{\frac{1}{\beta}}}{e^{\frac{1}{\beta}}+K-1} \\   \notag
        &= \frac{a(n-2)\beta \log K}{n} + \frac{b+1}{b+K}\cdot a\beta\log (b+K).
\end{align}

We can verify that $u_i(\s)\geq u_i(\s_j, \s_{-i})$ for any $2\leq j\leq n$ if we take 
\begin{equation}\label{eq:143}
    a=\beta\log K+1\leq\frac{\beta \log K+1}{\beta\Big(\frac{b+1}{b+K}\log (b+K)-\frac{1}{n}\log K\Big)}.
\end{equation}
This is because: 1. The inequality in \eqref{eq:143} always holds as $\frac{b+1}{b+K}\log (b+K) < \log (b+1)=\frac{1}{\beta}$ when $\beta\in [0,1]$ (this is due to the monotonicity of $\log x/x$); 2. $u_i(\s)= u_i(\s_j, \s_{-i})$ when $a=\frac{\beta \log K+1}{\beta\Big(\frac{b+1}{b+K}\log (b+K)-\frac{1}{n}\log K\Big)}$. Hence, $\s$ is an NE of $\G$. Putting \eqref{eq:125}, \eqref{eq:131}, and \eqref{eq:143} together, we have 
\begin{align}\notag
    PoA(\G) &=\frac{\max_{\s\in\S}W(\s)}{\min_{\al\in CCE(\G)}\EE_{\s\sim \al}[W(\s)]}\geq \frac{W(\s^*)}{W(\s)} \\ \label{eq:148}
    & = \frac{m\beta\log(b+K)}{n u_i(\s)} \geq  \frac{m}{n u_i(\s)} \\ \notag
    & = \frac{n+(n-1)a}{n(\beta\log K + 1) + a(n-1)\beta \log K} \\ \notag
    & =\frac{n-1}{n}+\frac{1-t^2a(a-1)}{a + ta(a-1)} \quad (t=\frac{n-1}{n}) \\ \label{eq:151}
    &> \frac{n-1}{n}+\frac{1}{5a-4} \\ \notag
    & =\frac{n-1}{n}+\frac{1}{1+5\beta\log K}.
\end{align}
where inequality \eqref{eq:148} holds because $\beta\log(b+K)\geq\beta\log(b+1)=1$, and \eqref{eq:151} holds because when $a=1+\beta\log K \in [1, 1.2]$ and $t=\frac{n-1}{n}\in [0.5,1)$, it is easy to verify that $\frac{1-t^2a(a-1)}{a + ta(a-1)} > \frac{1}{5a-4}$. It is equivalent to $4t^2a+4>5t^2a^2+ta$, which is true because  $t^2a(5a-4)+ta < a(5a-4)+a < 4$.
\end{proof}

\section{Proof of Proposition \ref{prop:unbound_user_welfare}}\label{app:unbound_user_welfare}
\begin{proof}
Consider a population of two users $\X=\{\x_1,\x_2\}$ and $n$ players in which one player has two pure strategies and the other $n-1$ players only have access to a single strategy, i.e, $\S_i=\{\s_0\},i=2,\dots,n$ and $\S_1=\{\s_1,\s_2\}$. Let the scoring function $\sigma$ be

\begin{equation}  
\sigma(\s, \x)=\left\{  
             \begin{array}{lr}  
             1, \text{~~if~~} \s=\s_1,\x=\x_1,&   \\ 
             \delta, \text{~~if~~} \s=\s_2,&   \\ 
             0, \text{~~otherwise.}&    
             \end{array}  
\right.  
\end{equation} 
 We will show that for any given $K\geq 1, 0\leq\beta\leq \min\{0.14, \frac{1}{5\log K} \}$ there exists $\delta \in (0,1)$ such that the PoA of game $\tilde{\G}(\{\S_i\}_{i=1}^n,\X,\sigma,\beta,K)$ is always strictly greater than $2$. 
 
 From the proof of Lemma \ref{lm:closed_form_utility}, the user utility and welfare functions of $\tilde{\G}$ share the same form as in \eqref{eq:user_utility}, \eqref{eq:welfare}, while the player utility functions of $\tilde{\G}$ have the following form

 \begin{equation}\label{eq:player_utility_G2}
        u_i(\s)= \sum_{j=1}^m \frac{\II[\s_i\in \T_j(\s;K)]\exp(\beta^{-1}\sigma(\s_i, \x_j))}{\sum_{\s_k\in \T_j(\s;K)} \exp(\beta^{-1}\sigma(\s_k, \x_j))}.
\end{equation}

Let $b=\exp(1/\beta)-1$ and we choose any $\delta\in [\delta_0,1)$ such that $\exp (\delta_0 / \beta) + K-1 = \frac{2}{\frac{1}{K}+\frac{1}{b+K}}$. Such $\delta_0\in (0,1)$ must exist because function $f(\delta)=\exp (\delta / \beta) + K-1$ is monotonically increasing in $[0,1]$ with range $[K, b+K] \supset \frac{2}{\frac{1}{K}+\frac{1}{b+K}}$. Given such choice of $\delta$, we can verify that
\begin{align}\notag
        2u_1(\s_2, \s_0, \cdots, \s_0) &= \frac{2\exp (\delta / \beta)}{\exp (\delta / \beta) + K-1}  \\\notag
        & \geq \frac{2\exp (\delta_0 / \beta)}{\exp (\delta_0 / \beta) + K-1}\\\notag
        &= \frac{\exp (1 / \beta)}{\exp (1 / \beta) + K-1} +\frac{1}{K} \\ \notag
        &\geq \frac{\exp (1 / \beta)}{\exp (1 / \beta) + K-1} +\frac{1}{n} = 2u_1(\s_1, \s_0, \cdots, \s_0),
\end{align}

which indicates that $(\s_2,\s_0, \cdots, \s_0)$ is a PNE of $\tilde{\G}$. Therefore, by picking $\delta=\delta_0$ we have
\begin{align}\notag
    PoA(\tilde{\G}) &=\frac{\max_{\s\in\S}W(\s)}{\min_{\al\in CCE(\tilde{\G})}\EE_{\s\sim \al}[W(\s)]}\geq \frac{W(\s_1, \s_0, \cdots, \s_0)}{W(\s_2, \s_0, \cdots, \s_0)} \\ \label{eq:381}
    &= \frac{\log(\exp (1/\beta)+K-1)+\log K}{2\log (\exp (\delta_0 / \beta) + K-1)} \\  \label{eq:382}
    & = \frac{\log(b+K)+\log K}{2\log [2K(b+K)] - 2 \log (b+2K)} & \mbox{by the choice of $\delta_0$}\\  \label{eq:383}
    & > 2,
\end{align}
where \eqref{eq:383} holds because \eqref{eq:383} is equivalent to 
\begin{equation}\label{eq:386}
    (b+2K)^4>16K^3(b+K)^3.
\end{equation}

And we show the correctness of \eqref{eq:386} by verifying the following situations:

\begin{enumerate}
    \item when $K\in \{2,3\}$, \eqref{eq:386} holds for all $\beta \in [0, 0.14]$, $b=\exp(1/\beta)-1$.
    \item when $K\geq 4$, from $\beta \leq \frac{1}{5\log K}$ we know $b+K=\exp(1/\beta)+K-1>K^5$ and thus 
    $\frac{(b+2K)^4}{(b+K)^3} > b > K^5 \geq 16 K^3$. Therefore, \eqref{eq:386} holds.
\end{enumerate}

Finally we show that when $K=1$ or $\beta \rightarrow 0$, $PoA(\tilde{\G})$ can be arbitrarily large.  

\begin{enumerate}
    \item when $\beta \rightarrow 0$, we have $b\rightarrow \infty$. From \eqref{eq:382} we have for any fixed $K$, 
    \begin{align}\notag
        \lim_{\beta\rightarrow 0}PoA(\tilde{\G}) &= \lim_{b\rightarrow +\infty} \Big\{\frac{\log(b+K)+\log K}{2\log [2K(b+K)] - 2 \log (b+2K)}\Big\} = \lim_{b\rightarrow +\infty}\log(b+K) \rightarrow +\infty.
    \end{align}
    
    \item when $K= 1$, the user's choice is deterministic and thus any $\delta\in (0,1)$ makes $(\s_2,\s_0,\cdots,\s_0)$ a PNE of $\tilde{\G}$. Let $\delta \rightarrow 0$ and from \eqref{eq:381} we have for any fixed $\beta$, 
    \begin{align}\notag
    \lim_{\delta\rightarrow 0}PoA(\tilde{\G}) &= \lim_{\delta\rightarrow 0} \Big\{ \frac{\log(\exp (1/\beta)+K-1)+\log K}{2\log (\exp (\delta / \beta) + K-1)}
    \Big\} = \lim_{\delta\rightarrow 0} \frac{1}{2\delta} \rightarrow + \infty.
    \end{align}
    
\end{enumerate}

\end{proof}

\section{Connections to Existing Models}\label{app:connection}

As an extended discussion to the related work, we show how our \game{}s connect to the following three previously proposed competition models for content creators. All the following models do not consider the presence of an RS and match each user with all content creators (players), which corresponds to the case $K=n$ in our setting. Interestingly, we found that each of them turns out to be a special case of our \game{}.

\vspace{2mm}
\noindent 
{\bf Facility location games under the no intervention mediator \citep{ben2019recommendation}}

Consider the following \game{} instance:
\begin{enumerate}
    \item the user population $\X\subseteq [0,1]$ is a finite set of size $m$,
    \item each player $i\in[n]$ shares the same action set $\S_i=[0,1]$,
    \item the scoring function is given by $\sigma(s, x)= |s-x|$,
    \item $(\beta, K)=(0, n)$,
    \item the utility function is induced by the user exposure metric, i.e., $u_i(\s)= \sum_{x\in \X}  \Pr(x \shortrightarrow s_i)$.

\end{enumerate}

If we let $m\rightarrow \infty$ so that $\X$ becomes a continuum with density function $g$ over the unit interval $[0,1]$, the game instance $\tilde{G}(\{\S_i\}_{i=1}^n, \X, \sigma, \beta, K)$ \footnote{Note that we use $\tilde{G}$ to refer to the variant of $G$ that utilizes the user exposure metric instead of the user engagement metric in player utility functions.} defined above is equivalent to the facility location game under the no intervention mediator proposed by \cite{ben2019recommendation}.

\vspace{2mm}
\noindent 
{\bf Hotelling-Downs model with limited attraction under support utility functions \citep{shen2016hotelling}}

Consider the following \game{} instance:
\begin{enumerate}
    \item the user population $\X=\{x_1,\cdots,x_m\}\subseteq [0,1]$ is a finite set of size $m$,
    \item each player $i\in[n]$ shares the same action set $\S_i=[0,1]\times[0,1]$. For each action $\s_i=(s_i,w_i)$ taken by player-$i$, it is associated with an attraction region $R_i=[s_i-\frac{w_i}{2}, s_i+\frac{w_i}{2}]\cap [0,1]$.
    \item for each $i\in [n]$, the scoring function is given by $\sigma(\s_i, x)= \II[x\in R_i]$,
    \item $(\beta, K)=(0, n)$,
    \item utility function is  the user engagement metric, i.e., $u_i(\s)= \sum_{j=1}^m \pi_j(\s) \Pr(x_j \shortrightarrow s_i)$.

\end{enumerate}

In fact, given $\beta=0$ and the above definition of $\sigma$, we can see the utility functions under both exposure and engagement metrics are identical, because it holds that $\pi_j(\s)\in \{0,1\}$ and $\pi_j(\s)=1$ if and only if $\Pr(x_j \shortrightarrow s_i)>0$. We can verify that the game instance $G(\{\S_i\}_{i=1}^n, \X, \sigma, \beta, K)$ defined above is equivalent to the Hotelling-Downs model with limited attraction under support utility functions proposed by \cite{shen2016hotelling}.

\vspace{2mm}
\noindent 
{\bf Exposure games  \citep{hron2022modeling}}

Consider the following \game{} instance:
\begin{enumerate}
    \item the user population $\X\subseteq \RR^d$ is a finite set of size $m$,
    \item each player $i\in[n]$ is associated with an action set $\S_i$ on the unit sphere in $\RR^d$, i.e., $\S_i\in \mathbb{S}^{d-1}$,
    \item the scoring function is given by the inner product, i.e., $\sigma(\s, \x)=\langle \s,\x\rangle $,
    \item $(\beta, K)=(\tau, n)$,
    \item the utility function is induced by the user exposure metric, i.e., $u_i(\s)= \sum_{x\in \X}  \Pr(x \shortrightarrow s_i)$.

\end{enumerate}

Note that in exposure games the parameter $\beta$ no longer represents the user decision noise but becomes a temperature parameter $\tau$ controlling the spread of exposure probabilities over items. The game instance $\tilde{G}(\{\S_i\}_{i=1}^n, \X, \sigma, \beta, K)$ defined above is equivalent to the exposure games proposed by \cite{hron2022modeling}.

\section{The Detailed Experimental Setup}\label{app:setup}

\subsection{The computation of globally optimal social welfare}

We use two heuristic methods for the computation of $\max_{\s\in\S}W(\s)$, when the exact solution is computationally infeasible: Simulated Annealing (SA) shown in Algorithm \ref{alg:SA} and Best-response Search (BRS) shown in Algorithm \ref{alg:BRS}. For SA, we set $T=5000$ and the temperature schedule $\tau_t=0.1/\sqrt{t}$; for BRS, we set $T=\max\{30, 2n\}$ and take the best output from 5 independent runs. Both methods perform as well as the brute-force search when the number of players $n$ and the size of action set $k$ are less than $6$, thus yielding the exact global optimal under such situations. For larger problem scales, we run both methods and select the best output to approximate $\max_{\s\in\S}W(\s)$.

\begin{algorithm}[h]
   \caption{Simulated annealing for computing the optimal welfare}
   \label{alg:SA}
\begin{algorithmic}
   \STATE {\bfseries Input:} Time horizon $T$, joint action space $\S=\prod_{i=1}^n\S_i$, welfare function $W(\s)$, temperature schedule $\{\tau_t\}_{t=1}^T$.
   \STATE {\bfseries Initialization:} A randomly selected joint action $\s^{(0)}=(\s_1^{(0)},\cdots,\s_n^{(0)})\in \S$.
   \FOR{$t=0$ {\bfseries to} $T$}
            \STATE Randomly choose a player $i\in [n]$ and randomly perturb her action $\s_i^{(t)}$ in $\s^{(t)}$ to yield $\s'^{(t)}$.
            \STATE Compute $W(\s'^{(t)}), W(\s^{(t)})$.
            \IF{$W(\s'^{(t)})> W(\s^{(t)})$} 
                \STATE Set $\s^{(t+1)}=\s'^{(t)}$.
            \ELSE
                \STATE With probability $e^{(W(\s'^{(t)}) - W(\s^{(t)}))/\tau_t }$, set $\s^{(t+1)}=\s'^{(t)}$; otherwise, $\s^{(t+1)}=\s^{(t)}$.
            \ENDIF
   \ENDFOR
   \STATE {\bfseries Output: $\max_{t\in [T]} W(\s^{(t)})$. }
\end{algorithmic}
\end{algorithm}

\begin{algorithm}[h]
   \caption{Best-response search for computing the optimal welfare}
   \label{alg:BRS}
\begin{algorithmic}
   \STATE {\bfseries Input:} Time horizon $T$, joint action space $\S=\prod_{i=1}^n\S_i$, welfare function $W(\s)$.
   \STATE {\bfseries Initialization:} A randomly selected joint action $\s^{(0)}=(\s_1^{(0)},\cdots,\s_n^{(0)})\in \S$.
   \FOR{$t=0$ {\bfseries to} $T$}
            \STATE Randomly choose a player $i\in [n]$ and search for her best-response that maximizes $W$, i.e., 
            $$\s_i^{(t+1)}=\arg\max_{\s_i\in\S_i} W(\s_i, \s_{-i}^{(t)}).$$
            \STATE Set $\s^{(t+1)}=(\s_i^{(t+1)}, \s_{-i}^{(t)})$.
   \ENDFOR
   \STATE {\bfseries Output: $W(\s^{(T)})$. }
\end{algorithmic}
\end{algorithm}

\subsection{The computation of the worst case welfare under CCE}


We can express the definition of CCE into a set of linear constraints. Let $\al \in$CCE$(\G)$ be a probability distribution over the joint action space $\S=\prod_{i=1}^n\S_i$. Then $\al$ satisfies
\begin{equation}\label{eq:lp_cce}
    \forall i, \forall \s'_i, \sum_{\s\in\S}\al(\s) u_i(\s) \geq \sum_{\s\in\S}\al(\s) u_i(\s'_i,\s_{-i}).
\end{equation}
And $\min_{\al\in CCE(\G)}\EE_{\s\sim \al}[W(\s)]$ is solving 
\begin{equation}\label{eq:lp_obj}
    \min_{\al\in CCE(\G)} \al(\s)W(\s)
\end{equation}
under linear constraints \eqref{eq:lp_cce}. Suppose each $\S_i$ shares the same size $k$, then we obtain a linear program with $k^n$ variables and $kn$ constraints.

\subsection{The details of no-regret dynamic simulation and computation of PotA}

For the computation of PotA, we need to simulate each player's sequence of play. We let all players run the following Exp-3 algorithm \ref{alg:Exp3} simultaneously to update their strategies in a fixed time horizon $T=5000$. According to \citep{auer2002nonstochastic}, Algorithm \ref{alg:Exp3} enjoys sublinear regret if $\eta=\epsilon\sim O(\sqrt{\frac{k\log k}{T}})$. However, it is not realistic to assume content creators in practice are sophisticated enough to figure out the game parameters and the optimal learning/exploration rates. Hence, unless specified, we always use a fixed value $(\eta, \epsilon)=(0.1,0.1)$ in our experiments. 

\begin{algorithm}[h]
   \caption{Exp3 for simulating player-$i$'s evolving strategies}
   \label{alg:Exp3}
\begin{algorithmic}
   \STATE {\bfseries Input:} Time horizon $T$, number of actions $k$, exploration parameter $\epsilon$, learning rate $\eta$.
   \STATE {\bfseries Initialization:} The score vector $\y_0=(y_1(0),\cdots, y_k(0))=(0,\cdots,0)$.
   \FOR{$t=0$ {\bfseries to} $T$}
   \STATE Compute a mixed strategy from the accumulated scores:
            $$
            p_j(t) = (1-\epsilon) \frac{ \exp (y_j(t)) }{ \sum_{l\in[k]} \exp(y_l(t)) } + \frac{\epsilon}{k},  \quad \forall j\in [k].
            $$
            \STATE Draw action $\s_{i,t}\in[k]$ randomly accordingly to the distribution $\p_t=(p_1(t), p_2(t), \dots p_k(t))$.
            \STATE Play action $\s_{i,t}$ and observe the utility $u_{i}(\s_{i,t},\s_{-i,t})$.
            \STATE Update the score
            $$
            y_{\s_{i,t}}(t+1) =  y_{\s_{i,t}}(t) + \eta\cdot\frac{u_{i}(\s_{i,t},\s_{-i,t})}{p_{\s_{i,t}}(t)}.
            $$
   \ENDFOR
\end{algorithmic}
\end{algorithm}

\end{document}